\documentclass[11pt,a4paper]{article}
\usepackage{theorem}
\usepackage[T1]{fontenc}
\usepackage{amssymb, amsmath}
\usepackage{graphicx}
\usepackage{xspace}
\usepackage{makecell}
\usepackage{hyperref}
\usepackage{authblk}

 \newtheorem{theorem}{Theorem}
 
 \theorembodyfont{\upshape}
 \newtheorem{lemma}[theorem]{Lemma}
 \newtheorem{proposition}[theorem]{Proposition}
 \newtheorem{corollary}[theorem]{Corollary}
 \theoremstyle{plain}
 \newtheorem{definition}[theorem]{Definition}
 \newtheorem{remark}[theorem]{Remark}
 
\newtheorem{example}[theorem]{Example}

 \def\koniec{{}\hfill$\square$}
 \newenvironment{proof}[1][!*!,!]%
 {\noindent{\emph{Proof.} }}%
 {\koniec\medskip}

\title{On parametric verification of
  asynchronous, shared-memory pushdown systems} 

\author[1]{Marie Fortin}
\author[2]{Anca Muscholl}
\author[3]{Igor Walukiewicz}
\affil[1]{ENS Cachan, University of Paris-Saclay}
\affil[2]{Technical University of Munich, IAS\footnote{On leave from University of Bordeaux, LaBRI.}}
\affil[3]{University of Bordeaux, CNRS, LaBRI}

\date{}

 \newcommand{\Nat}{\ensuremath{\mathbb{N}}}

\newcommand{\es}{\emptyset}
\newcommand{\imp}{\Rightarrow}

\newcommand{\incl}{\subseteq}

\newcommand{\ch}{\mathit{choose}}

\newcommand{\sq}[1]{[#1]}
\newcommand{\di}[1]{\langle #1 \rangle}
\newcommand{\sqq}[1]{\sq{\cdot }}
\newcommand{\ddi}[1]{\di{\cdot }}
\newcommand{\set}[1]{\{#1\}}

\renewcommand\bar[1]{\overline{#1}}

\newcommand{\act}[1]{\stackrel{#1}{\longrightarrow}}

\renewcommand{\a}{\alpha}
\renewcommand{\b}{\beta}
\renewcommand{\d}{\delta}
\newcommand{\e}{\varepsilon}

\renewcommand{\k}{\kappa}

\newcommand{\n}{\nu}

\renewcommand{\r}{\rho}

\newcommand{\w}{\omega}

\newcommand{\D}{\Delta}

\newcommand{\G}{\Gamma}

\renewcommand{\S}{\Sigma}

\newcommand{\Aa}{\mathcal{A}}

\newcommand{\Cc}{\mathcal{C}}
\newcommand{\Dd}{\mathcal{D}}

\newcommand{\Pp}{\mathcal{P}}

\newcommand{\PSPACE}{\text{\sc Pspace}}

\newcommand{\NP}{\text{\sc NP}}
\newcommand{\coNP}{\text{\sc coNP}}
\newcommand{\EXPTIME}{\text{\sc Exptime}}
\newcommand{\NEXPTIME}{\text{\sc Nexptime}}
\newcommand{\coNEXPTIME}{\text{co\sc{Nexptime}}}
\newcommand{\EXPSPACE}{\text{\sc Expspace}}

\renewcommand{\paragraph}[1]{\smallskip\par\noindent {\bf #1}}

\newcommand{\struct}[1]{\langle #1 \rangle}

\def\sqr#1#2{\vbox
 {\hrule height#2
  \mbox{\vrule width#2 height#1 \kern#1 \vrule width#2}%
  \hrule height#2}}

\newenvironment{proofof}[1]{\noindent\textit{Proof of {#1}. }}{\nopagebreak
  \hspace*{\fill}$\Box$\medskip\par}
  \newenvironment{proofsketch}{\noindent\textit{Sketch of proof. }}{}

{\end{list}}

{\end{list}}

{\end{list}}

\newcount\hour \hour=\time
\newcount\minute \minute=\hour
\divide\hour by 60
\newcount\temp \temp=\hour
\multiply\temp by 60
\advance\minute by -\temp
\newcount\temp \temp=\minute
\divide\temp by 10
\def\optzero{\ifcase\temp 0\fi}

\def\miesiaca{\ifcase\month\or stycznia\or lutego\or marca\or kwietnia\or maja
                           \or czerwca\or lipca\or sierpnia\or wrze\'snia
                           \or pa\'zdziernika\or listopada\or grudnia\fi}
 
\newcommand{\Cck}{\Cc^\k}
\newcommand{\Ddk}{\Dd^\k}

\newcommand{\xra}{\xrightarrow}
\newcommand{\dr}{r}
\newcommand{\dw}{w}
\newcommand{\Dr}{r}
\newcommand{\Dw}{w}
\newcommand{\Cr}{\bar r}
\newcommand{\Cw}{\bar w}
\newcommand{\CDsystem}{$(\Cc,\Dd)$-system\xspace}
\newcommand{\CDsystems}{$(\Cc,\Dd)$-systems\xspace}
\newcommand{\CDpsystem}{$(\Cc',\Dd')$-system\xspace}
\newcommand{\CDpsystems}{$(\Cc',\Dd')$-systems\xspace}
\newcommand{\CDbsystem}{$(\Cc'',\Dd'')$-system\xspace}

\newcommand{\CDtsystem}{$(\tCc,\tDd)$-system\xspace}

\newcommand{\init}{{\mathit{init}}}
\newcommand{\rest}{\!\downarrow\,}

\newcommand\Clsupp{$\omega$-supported\xspace}

\newcommand\cs[1]{[#1]}
\newcommand\nst[2]{#1^{(#2)}}
\newcommand{\shuffle}{\mathbin{{\sqcup}\mathchoice{\mkern-3mu}{\mkern-3mu}{\mkern-3.4mu}{\mkern-3.8mu}{\sqcup}}}

\newcommand{\sleq}{\sqsubseteq}

\newcommand{\lst}{\mathit{last}}
\newcommand{\ttop}{\mathit{top}}

\newcommand{\LE}{\mathsf{E}}
\newcommand{\LF}{\mathsf{F}}
\newcommand{\LA}{\mathsf{A}}
\newcommand{\LG}{\mathsf{G}}

\newcommand\Qc{P}
\newcommand\Gc{\G_C}
\newcommand\Sc{\S_C}
\newcommand\dc{\delta}
\newcommand\qic{p_\init}
\newcommand\Aic{A_\init^C}

\newcommand\Qcf{S}

\newcommand\qicf{s_\init}

\newcommand\Ccf{\Cc_{\mathit{fin}}}
\newcommand\Ccfk{\Cc_{\mathit{fin}}^\k}
\newcommand\Qn{P_{\mathit{fin}}}
\newcommand\qin{p^{\mathit{fin}}_\init}

\newcommand\dn{\delta}

\newcommand\Qd{Q}
\newcommand\Gd{\G_D}
\newcommand\Sd{\S_D}
\newcommand\dd{\Delta}
\newcommand\qid{q_\init}
\newcommand\Aid{A_\init^D}

\newcommand{\tCc}{\tilde{\Cc}}
\newcommand{\tDd}{\tilde{\Dd}}
\newcommand{\tPp}{\tilde{\Pp}}
\newcommand{\tSd}{\widetilde{\S}_D}

\newcommand\qcw[1]{\overline{\textup{\texttt{w?}}}(#1)}
\newcommand\qcr[1]{\overline{\textup{\texttt{r?}}}(#1)}
\newcommand\acw[1]{\overline{\textup{\texttt{w}}}(#1)}
\newcommand\acr[1]{\overline{\textup{\texttt{r}}}(#1)}

\newcommand\gdw[1]{\textup{\texttt{w}}(#1)}
\newcommand\gdr[1]{\textup{\texttt{r}}(#1)}

\newcommand{\trans}{\mathit{trans}}
\newcommand{\sexp}{\mathit{stutt}}

\newcommand\Scp{\S'_C}
\newcommand\Sdp{\S'_D}

\newcommand\Sck{\Sc^m}
\newcommand\Sdk{\Sd^m}

\newcommand\qcwi[1]{\overline{\textup{\texttt{w}}_{\textup{\texttt i}} \textup{\texttt{?}}}(#1)}
\newcommand\qcri[1]{\overline{\textup{\texttt{r}}_{\textup{\texttt i}} \textup{\texttt{?}}}(#1)}
\newcommand\acwi[1]{\overline{\textup{\texttt{w}}}_{\textup{\texttt i}}(#1)}
\newcommand\acri[1]{\overline{\textup{\texttt{r}}}_{\textup{\texttt i}}(#1)}
\newcommand\gdwi[1]{\textup{\texttt{w}}_{\textup{\texttt i}}(#1)}
\newcommand\gdri[1]{\textup{\texttt{r}}_{\textup{\texttt i}}(#1)}

\begin{document}

\maketitle

\begin{abstract}
  We consider the model of parametrized asynchronous shared-memory
pushdown systems as introduced in [Hague'11].  In a series of recent papers
it has been shown that reachability in this model is \PSPACE-complete
[Esparza, Ganty, Majumdar'13] and that liveness is
decidable in \NEXPTIME\ [Durand-Gasselin, Esparza, Ganty, 
Majumdar'15]. We show that the liveness problem is
\PSPACE-complete. We also consider the universal reachability problem.  We
show that it is decidable, and \coNEXPTIME-complete.  
Finally, using these results, we prove that verifying general
regular properties of traces of executions, satisfying some stuttering
condition, is also decidable in \NEXPTIME\ for this model.
\end{abstract}

\section{Introduction}
It is common knowledge that even boolean programs may be impossible to
analyze algorithmically.  Features such as recursion or parallelism
make the set of reachable configurations potentially infinite.  The
usual example is given by systems consisting of two pushdown processes
with a shared boolean variable.  Such a model can simulate a Turing
machine, so every non-trivial question is undecidable~\cite{ram00}.
Kahlon~\cite{kah08} proposed to consider a parametric version of this
model, where the number of pushdown processes is arbitrary.  At first
sight this may look like a more general model, but it turns out that
model-checking under various synchronization primitives is decidable,
due to the lack of process identities.

Later Hague~\cite{hag11fsttcs} considered a model where a single process
has an identity, the leader process, but the operations on the shared
variable do not involve any synchronization. His model of parametrized asynchronous
shared-memory pushdown systems of~\cite{hag11fsttcs} consists of
one leader process and an arbitrary number of identical, anonymous contributor
processes.  Processes communicate through a shared, bounded-value
register using write and read operations.  A very important aspect
is that there are no locks, nor
test-and-set type operations, since this kind of operations would
allow to elect a second leader, and all questions would be immediately
undecidable.  The main result
of~\cite{hag11fsttcs} is that this model still
enjoys a decidable reachability problem.  
The complexity of this problem has been later established in~\cite{EGM16}.
Recently,
Durand-Gasselin et al.~\cite{DEGM15} have also shown decidability of the
liveness problem for this model.



The  reachability problem in Hague's model
can be formulated as whether there is a
computation of the system where the leader
can execute a special action, say $\top$.
The \emph{repeated reachability} problem asks if there is a
computation where the leader can execute $\top$ infinitely often.
This problem provides a succinct way of talking about
liveness properties concerning the leader process. 
We also consider in this paper \emph{universal reachability}: this is
the question of
deciding if on every maximal trace of the system, the leader executes $\top$.
In terms of temporal logics, reachability is about $\LE\LF$ properties,
while universal reachability is about $\LA\LF$ properties.


Our first result shows that there is no complexity gap between
verification of reachability and repeated reachability in the
parametrized setting, both problems are 
\PSPACE-complete.  This answers the question left open by~\cite{DEGM15},
that provided a \PSPACE{} lower bound and a \NEXPTIME{} upper bound
for the liveness problem. 
Technically, our
\PSPACE\ upper bound requires to combine the techniques from~\cite{DEGM15}
and~\cite{LMW15}.  
We use a result from~\cite{DEGM15} saying
that if there is a run then there is an ultimately periodic one.  Then
we extend the approach from~\cite{LMW15} from finite to ultimately
periodic runs. 

As a second result we show that universal reachability is
\coNEXPTIME-complete. This result
also  bears some interesting technical aspects.  For the upper bound, as in the
case of reachability, we use a variant of the so-called
accumulator semantics~\cite{LMW15,DEGM15}.  We need to adapt it though
in order to make it sensitive to divergence.  The lower bound shows
that it is actually possible to force a fair amount of synchronization
in the model; we can ensure that the first $2^n$ values written into the
shared register are read in the correct order and none of them is
skipped.  So the \coNEXPTIME-hardness result can be interpreted
positively, as showing what can be implemented in this model.

Finally, we consider properties that refer not only to the leader
process, but also to contributors. As noticed in \cite{DEGM15} such properties are undecidable  in
general, because they can enforce
special interleavings that amount to identify a particular
contributor. In the parametrized
setting it is more natural to use properties that are stutter-invariant
w.r.t.~contributor actions. We show that parametrized verification of such
properties is decidable, and establish precise complexities
both when the property is given as a B\"uchi automaton,
and as LTL formula.
\medskip

\noindent \emph{Related work.}
Parametrized verification of shared-memory,
multi-threaded programs has been studied for finite-state threads e.g.~in
\cite{BCR01,KKW10} and for pushdown threads in
\cite{ABQ11,BESS05,LMP10,LMP12}. 
The decidability results in
\cite{ABQ11,BESS05,LMP10,LMP12} 
concern the reachability
analysis up to a bounded number of execution contexts, and in
\cite{ABQ11,BESS05}, dynamic thread creation is allowed.
The main difference with our
setting is that synchronization primitives are allowed in those
models, so decidability depends on restricting the set of executions
in the spirit of bounded context switches.
Our model does not have such a restriction, but forbids synchronization
instead. 

Besides the already cited papers, a related paper that goes beyond
the reachability property is~\cite{BMRSS-icalp16}. Bouyer et.~al.~consider
in~\cite{BMRSS-icalp16} 
a very similar setting, but without leader and only finite-state
contributors.
They consider the problem of \emph{almost-sure} reachability, which asks
if a given state is reached by some process with probability 1 under a
stochastic scheduler. 
They exhibit the existence of positive or negative cut-offs, and show
that the problem can be decided in \EXPSPACE, and is at least \PSPACE-hard.
By contrast, the universal reachability problem we consider here, 
although at first glance close to
the question in~\cite{BMRSS-icalp16}, has very different characteristics.
It is in \NP\ for 
finite-state contributors, and can be simply solved by using a variant
of the accumulator semantics.
The challenge in the present paper comes from considering pushdown systems.

Finally, we should mention that there is a rich literature concerning the verification of
\emph{asynchronously-communicating} parametrized  programs, that is
mostly related to the verification of distributed protocols and  
varies  for approaches and models (see
e.g.~\cite{germansistla92,CGJ97,EsparzaFM99lics,KPSZ02} for some early
work, and \cite{Delzanno14,NT15,veith15book} and references
therein). Most of these papers are concerned with finite-state
programs only, which are not the main focus of our results.
\medskip

\noindent \emph{Outline of the paper.}
In Section~\ref{sec:prelim}, we introduce Hague's model.
Section~\ref{sec:problem} presents the problems considered in this paper, and
gives an overview of our results.
In Section~\ref{sec:repeated} we study the repeated reachability problem,
and in Section~\ref{sec:universal} universal reachability.
Finally, we show in Sections~\ref{sec:general} and \ref{sec:simplify} how
these results may be used to verify more general properties. 


\section{Preliminaries}
\label{sec:prelim}

\synctex=1

In this section we recall the model of parametrized systems
of~\cite{hag11fsttcs}, which we call
\CDsystems. First we give some basic definitions and notations.

 \subsection{Standard definitions}

A \emph{multiset} over a set $E$ is a function $M: E \to \Nat$.
We let $|M| = \sum_{x \in E} M(x)$.
The \emph{support} of $M$ is the set $\set{x \in E \mid M(x) > 0}$.
For $n \in \Nat$, we write $nM$, $M + M'$ and $M - M'$ for the multisets 
defined by $(nM)(x) = n \cdot M(x)$, $(M + M')(x) = M(x) + M'(x)$ and 
$(M - M')(x) = \max(0,M(x) - M'(x))$.
We denote by $[x]$ the multiset containing a single copy of $x$, 
and $[x_1, \ldots, x_n]$ the multiset $[x_1] + \ldots + [x_n]$.
We write $M \le M'$ when $M(x) \le M'(x)$ for all~$x$.

A \emph{transition system} over a finite alphabet $\S$ is a tuple 
$\struct{S,\d,s_\init}$ where $S$ is a (finite or infinite) set of states, 
$\d \subseteq S \times \S \times S$ is a set of transitions, and
$s_\init \in S$ the initial state. We write $s \xra{u} s'$ (for $u \in \S^*$)
when there exists a path from $s$ to $s'$ labeled by $u$.
A \emph{trace} is a sequence of actions labeling a path starting in
$s_\init$; so $u$ is a trace if $s_\init\xra{u} s'$ for some $s'$.

A \emph{pushdown system} is a tuple 
$\struct{Q,\S,\Gamma,\D,q_\init,A_\init}$ consisting of a finite set of states $Q$,
a finite input alphabet $\S$, a finite stack alphabet $\Gamma$, a set of transitions 
$\D \subseteq (Q\times\G) \times (\S\cup\{\e\}) \times (Q\times\G^*)$,
an initial state $q_\init \in Q$, and an initial stack symbol $A_\init \in \G$.
The associated transition system has $Q \times \G^*$ as states, $q_\init A_\init$
as the initial state, and transitions $qA\a \xra{a} q'\a'\a$ for
$(q,A,a,q',\a') \in \D$.

A word $u = a_1 \cdots a_n$ is a \emph{subword} of $v$ (written $u \sleq v$)
when there are words $v_0, \ldots, v_{n}$ such that 
$v = v_0 a_1 v_1 \cdots v_{n-1} a_n v_n$, so $u$ is obtained from $v$
by erasing symbols.
The \emph{downward closure} of a language $L \subseteq \S^*$ is 
$L \rest = \set{u \in \S^* \mid \exists v \in L.\, u \sleq v}$.

\subsection{\CDsystems}\label{sec:CD}

We proceed to the formal definition of \CDsystems. These systems are
composed of arbitrary many instances of a contributor process $\Cc$
and one instance of a leader process~$\Dd$. The processes
communicate through a shared register. We write $G$ for the finite set
of register values, and use $g,h$ to range over elements of $G$. The
initial value of the register is denoted by $g_\init$.  The alphabets of
both $\Cc$ and $\Dd$ contain actions representing reads and writes to
the register:
\begin{equation*}
  \Sc=\set{\Cr(g),\Cw(g) : g\in G}\ \qquad 
  \Sd=\set{\dr(g),\dw(g) : g\in G}\, .
\end{equation*}
Both $\Cc$ and $\Dd$ are (possibly infinite) transition
systems over these alphabets: 
\begin{equation}\label{eq1}
  \Cc=\struct{S,\d \subseteq S\times\Sc\times S,s_\init}\qquad
  \Dd=\struct{T,\D \subseteq T\times\Sd\times T, t_\init}
\end{equation}
In this paper we will be interested in the special case where $\Cc$ and
$\Dd$ are pushdown transition systems:
\begin{equation}\label{eq2}
  \Aa_C=\struct{\Qc,\Sc,\Gc,\dc,\qic,\Aic}\qquad
  \Aa_D=\struct{\Qd,\Sd,\Gd,\dd,\qid,\Aid}
\end{equation}
In this case the transition system $\Cc$ from~\eqref{eq1}
is  the transition system associated  with $\Aa_C$: its set of states is 
$S=\Qc\times(\Gc)^*$ and the transition relation $\d$ is defined by
the push and pop operations. 
Similarly, the transition system $\Dd$ is determined by $\Aa_D$.
When stating general results on \CDsystems we will use the notations
from~Eq.~\eqref{eq1}; when we need to refer to precise states, or use
some particular property of pushdown transition systems, we will 
 employ the notations from~Eq.~\eqref{eq2}.

A $(\Cc,\Dd)$-system consists of an arbitrary number of copies of
$\Cc$, one copy of $\Dd$, and a shared register. So a
\emph{configuration} is  a triple $(M \in \Nat^S,t\in
T,g\in G)$, consisting of a multiset $M$ counting the number of
instances of $\Cc$ in a given state,  the state $t$ of $\Dd$ and 
the current register value~$g$.

In order to define the transitions of the \CDsystem{} we extend the
transition relation $\d$ of $\Cc$ from elements of $S$ to multisets
over~$S$:
\begin{align*}
  M\xra{a}M' \text{ in $\d$\qquad} & \text{if $s \xra{a} s'$ in $\d$, $M(s) > 0$, 
and $M' = M - [s] + [s']$}, \\ & \text{for some $s,s' \in S$.}
\end{align*}
Observe that such a transition does not change the size of the multiset.
The transitions of the \CDsystem are either transitions of the leader 
(the first two cases below) or transitions  of contributors (last two cases):
\begin{align*}
  (M,t,g)\xra{\dw(h)}& (M,t',h)&& \text{if $t\xra{\dw(h)} t'$ in
                                  $\D$}\,, \\
  (M,t,g)\xra{\dr(h)}& (M,t',h)&& \text{if $t\xra{\dr(h)} t'$ in
                                  $\D$ and $h=g$}\,,\\
  (M,t,g)\xra{\Cw(h)}& (M',t,h)&& \text{if $M\xra{\Cw(h)} M'$ in $\d$} \,,\\
  (M,t,g)\xra{\Cr(h)}& (M',t,h)&& \text{if $M\xra{\Cr(h)} M'$ in $\d$ and $h=g$}\,.
\end{align*}
A \emph{run} from a configuration $(M,t,g)$ is a finite or an infinite
sequence of transitions starting in $(M,t,g)$. 
A run can start with any number $n$ of contributors, but then the
number of contributors is constant during the run.
A run is \emph{initial} if it starts in a configuration
of the form $(n[s_\init],t_\init,g_\init)$, for some $n \in \Nat$.
It is \emph{maximal} if it is initial and is not a prefix of any other run.
In particular, every infinite initial run is maximal.
A \emph{(maximal) trace} of the \CDsystem is a finite or an infinite sequence 
over $\Sc \cup \Sd$ labeling a (maximal) initial~run.



\section{Problem statement and overview of results}\label{sec:problem}

We are interested in
verifying properties of traces of \CDsystems.
The first question is what kind of specifications we may use. 
In general, verification of regular, action-based properties of runs of
pushdown \CDsystems is undecidable: such properties  allow to control
the interleavings of contributors and thus identify e.g.~a single
contributor that runs together with the leader (see
e.g.~\cite{DEGM15}).

For this reason we  consider 
\emph{$\Cc$-expanding} properties $\Pp \subseteq (\Sc \cup
\Sd)^\infty$.  By this we mean properties where actions of
contributors can be replicated: if $u_0 a_0 u_1
a_1 u_2 \cdots \in \Pp$ with $a_i\in\Sc$, $u_i \in \Sd^*$, and $f
:\Nat \to \Nat^+$, then $u_0 a_0^{f(0)} u_1 a_1^{f(1)} u_2 \cdots \in
\Pp$, too. Because of parametrization, contributor actions can 
always be replicated, so $\Cc$-expanding properties are a natural
class of properties for \CDsystems.

A related, more classical notion is \emph{stutter-invariance}. A
language $L \subseteq \S^\infty$ is stutter-invariant if for every
finite or infinite sequence $a_0 a_1 \cdots$ and every function $f
:\Nat \to \Nat^+$, we have $a_0 a_1 \cdots \in L$ iff $a_0^{f(0)}
a_1^{f(1)} \cdots \in L$. It is known that the stutter-invariant
properties expressible in linear-time temporal logic LTL are
precisely those expressible in LTL without the
next-operator~\cite{PW97,Ete00}.
By definition, every stutter-invariant property is $\Cc$-expanding.

We will consider regular properties $\Pp\subseteq (\Sc \cup
\Sd)^\infty$ that are $\Cc$-expanding,
as defined above. 
These properties will be described either by an LTL formula, or by an automaton
$\Aa=\struct{Q,\Sc \cup \Sd,\D,q_0,F,R}$ with finite set of states $Q$,
alphabet $\Sc \cup \Sd$, 
transitions $\D \subseteq Q \times (\Sc \cup \Sd) \times Q$, and
sets of final states $F,R \subseteq Q$. Finite runs are accepted if
they end in $F$, infinite ones are accepted if they visit $R$ infinitely often.
For simplicity, we call such an automaton a \emph{B\"uchi automaton}.

We will also consider some particular properties, that are essential for the 
general decision procedure:

\begin{enumerate}
\item The \emph{reachability problem} asks if the \CDsystem has some trace 
containing a given leader action $\top$. 
\item The \emph{repeated reachability problem} asks if the
  \CDsystem has some 
trace with infinitely many occurrences of a given leader action $\top$.
\item The \emph{universal reachability problem} asks if every maximal
  trace of the  \CDsystem 
  contains a given leader action $\top$. 
\item The complement of the previous question is the \emph{max-safe
    problem}. It asks if the  \CDsystem has some
  maximal trace that does not contain a given leader action $\top$.
\end{enumerate}

The above problems are of course basic examples of (stutter-invariant) LTL
properties over $\Sd \cup \Sc$. 
In the following example we show how they can be used together to
verify some more involved properties of parametrized systems.

\begin{example}
  The consensus problem consists in making all processes agree on a common
  value, among those values that were proposed by the processes. For
  simplicity we can assume that each (leader or contributor) process
  selects initially some value $b\in\set{0,1}$, and has two special actions,
  $\ch(0)$ and $\ch(1)$. If a (leader or contributor) process performs
  $\ch(i)$, this means that a value has been agreed upon, and it is
  equal to $i$. Furthermore, we denote the set of actions that are
  possible after $\ch(i)$ as $\S_i$, and assume that $\S_0 \cap
  \S_1=\es$.

The property we are interested in asks that processes should agree
on a common value from $\set{0,1}$, and then use this value in all
future computations. 
Stated as a property of the leader and contributors it means:
\[
\LA\LF(\bigvee_{b=0,1} \ch(b)) \wedge \LA\LG (\bigwedge_{b=0,1}
(\ch(b) \imp \LA\LG \, \S_b))
\]
The first part of the property corresponds to universal
reachability. The second one is a safety property (the complement of
a reachability property), requiring that
there is no maximal run containing action $\ch(b)$ and later on, some
action from $\S_{1-b}$. 
\end{example}


The main results of our paper establish the precise complexity of all
questions about \CDsystems that were introduced above.
We first state the result concerning the largest class of properties,
and then for specific cases. 
The proofs follow the inverse order, the results about specific cases
are used to prove the general result.

\begin{theorem}\label{th:main}
 The following problem is \NEXPTIME-complete:
given a pushdown \CDsystem and a $\Cc$-expanding regular
property $\Pp$ over $\Sd \cup \Sc$ (given by a B\"uchi automaton
or by an LTL formula), determine if the \CDsystem has a
maximal trace in $\Pp$.
\end{theorem}
In principle, it could be more difficult to verify properties given by
LTL formulas since LTL formulas can be exponentially more succinct
than nondeterministic B\"uchi automata.  The above theorem implies
that this blowup in the translation from LTL to non-deterministic
B\"uchi automata does not influence the complexity of the algorithm.
In this context, it is worth to recall that even for a single
pushdown automaton, LTL model-checking
is \EXPTIME-complete~\cite{BEM97}.

\begin{theorem}\label{thm:liveness}
  The repeated reachability problem for pushdown \CDsystems is
  \PSPACE-complete.
\end{theorem}

\begin{theorem}\label{th:safety}
  The max-safe problem for pushdown \CDsystems is \NEXPTIME-complete. It
  is \NP-complete when $\Cc$ ranges over finite-state systems.
\end{theorem}



The proof of Theorem~\ref{th:main} uses Theorems~\ref{thm:liveness}
and~\ref{th:safety} and one more result, that is interesting on
its own. 
Note that both the repeated reachability and
the max-safe 
problem talk about one distinguished action of the leader, while
$\Cc$-expanding properties also refer to actions of
contributors.
Perhaps a bit surprisingly, we show how to modify a \CDsystem
so that only leader actions matter.
We reduce the problem of verifying if  a \CDsystem satisfies a
$\Cc$-expanding property $\Pp$ to the problem of verifying that some
polynomially larger \CDtsystem satisfies a property $\tPp$ that refers \emph{only} to actions
of the leader $\tDd$.
The idea is that the leader is given the control of the register, and
contributors just submit read or write requests; these are processed
by the leader who later  sends acknowledgements to the contributors. 
The result of this reduction is summarized by the following theorem:

\begin{theorem}
  \label{th:transP}
For every \CDsystem, there exists a \CDtsystem such that for every 
$\Cc$-expanding property $\Pp \subseteq (\Sd \cup \Sc)^\infty$, 
there exists a property $\tPp \subseteq (\tSd)^\infty$, where $\tSd$ is 
the action alphabet of $\tDd$, such that:
\begin{enumerate}
\item   the \CDsystem has a finite (resp.~infinite) maximal trace in $\Pp$ iff 
  the \CDtsystem has a finite (resp.~infinite) maximal trace whose projection 
  on $\tSd$  is in $\tPp$;
\item every infinite run of the \CDtsystem has infinitely many write
  operations of $\tDd$;
\item the \CDtsystem has an infinite run iff the \CDsystem has one.
\end{enumerate}
If the \CDsystem is defined by pushdown automata $\Aa_C$ and $\Aa_D$, then
the \CDtsystem is effectively defined by pushdown automata of sizes linear
in the sizes of $\Aa_C$ and $\Aa_D$.
If $\Pp$ is a regular, respectively LTL property, then so is $\tPp$. 
An automaton or LTL formula of linear size for $\tPp$ is effectively
computable from the one for $\Pp$.
\end{theorem}

In the remaining of the paper we successively give the proofs of 
Theorem~\ref{thm:liveness}, Theorem~\ref{th:safety}, Theorem~\ref{th:main}, 
and Theorem~\ref{th:transP}.
Though it is proven in the last section, Theorem~\ref{th:transP} does not rely 
on other results, and will be used in the proofs of Theorem~\ref{th:safety}
and \ref{th:main}.



\section{Repeated reachability}
\label{sec:repeated}

We show in this section that repeated reachability for pushdown \CDsystems can
be decided in
\PSPACE{}
(Theorem~\ref{thm:liveness}). 
The matching lower bound comes from the \PSPACE\ lower bound for the
reachability problem~\cite{EGM16}.
We call a run of the \CDsystem\ a \emph{B\"uchi run} if it has
infinitely many occurrences of the leader action $\top$.  
So the problem is to decide if a given \CDsystem has a B\"uchi run. 

Our proof has three steps.
The first one relies on a result
from~\cite{DEGM15}, showing that the stacks
of  contributors can be assumed to be polynomially bounded.
This allows to search for ultimately periodic
runs  (Lemma~\ref{lem:loop}), as in the case of one pushdown system. 
The next step extends the capacity technique introduced
in~\cite{LMW15} for the reachability problem, to B\"uchi runs.
We reduce the search  for B\"uchi runs to the existence of  $\omega$-supported
runs (Lemma~\ref{lem:Clsupp}). 
The last step is the observation that, as in the case of finite runs,
we can use the downward closure of capacity runs of the leader
(Lemma~\ref{lem:subwords}).
Overall this yields a \PSPACE\ algorithm for the existence of B\"uchi
runs (Theorem~\ref{thm:liveness}).

\subsection{Finite-state contributors} 

As observed in~\cite{EGM16}, pushdown contributors 
can be simulated by finite-state ones, by exploiting the fact that the
setting is parametrized. To state the result we need 
the notion of~\emph{effective stack-height} for a pushdown system.
\emph{Effective} stack-height refers to the part of the stack that
is still used in the future.
Consider a possibly infinite run
$\rho = q_1\a_1 \xra{a_1} q_2\a_2 \xra{a_2} \ldots$ of a pushdown
system.  We write $\a_i = \a'_i \a''_i$, where $\a''_i$ is the longest
suffix of $\a_i$ that is also a proper suffix of $\a_j$ for all $j>i$.
The \emph{effective stack-height} of a configuration $q_i\a_i$ in $\rho$
is the length of $\a'_i$.  (Notice that even though it is never
popped, the first element of the longest common suffix of the
$(\a_i)_{j\ge i}$ may be read, hence the use of \emph{proper}
suffixes.)

By $\Cc_N$ we denote the restriction of
the contributor pushdown $\Aa_C$ to runs in which all configurations have effective stack-height
at most $N$, where $N$ is a positive integer.
More precisely, $\Cc_N$ is the finite-state system with set of states
$\set{p\a \in \Qc\Gc^* : |\a| \le N}$, and transitions 
$p\a \xra{a} q\a'$ if $p\a \xra{a} q\a'\a''$ in $\D$ for some $\a''$.
Note that $\Cc_N$ is effectively computable in \PSPACE{} from $\Aa_C$ and $N$ given
in unary.
The key idea in~\cite{DEGM15} is that when looking for
B\"uchi runs for pushdown \CDsystems, $\Cc$ can be replaced by $\Cc_N$ for $N$ 
polynomially bounded:

\begin{theorem}[Thm.~4 in~\cite{DEGM15}]\label{th:buechi}
  Let $N > 2|\Qc|^2|\Gc|$. 
  There is a B\"uchi run in the \CDsystem\ iff there is one in the 
  $(\Cc_N,\Dd)$-system.
\end{theorem}

A similar result for finite runs can be derived from the proof of this
theorem.

\begin{lemma}
  \label{lem:Ccf}
  Let $N > 2|\Qc|^2|\Gc|+1$.
  A configuration $([p_1\a_1,\ldots,p_n\a_n],t,g)$ of the
  $(\Cc_N,\Dd)$-system is reachable iff there exists a reachable configuration 
  of the \CDsystem of the form
  $([p_1\a_1\b_1,\ldots,p_n\a_n\b_n],t,g)$, for some $\b_i$.
\end{lemma}

\noindent \emph{Notation.} Throughout the paper, we write $\Ccf$
for $\Cc_N$ with $N= 2|\Qc|^2|\Gc|+2$.
We will use the notation $\struct{\Qn,\Sc,\dn,\qin}$ for the
finite-state system~$\Ccf$, and continue to write
$\Aa_D=\struct{\Qd,\Sd,\Gd,\dd,\qid,\Aid}$ for the pushdown system
$\Dd$.

\begin{remark}
  \label{lem:est1}
  It is not difficult to see that every infinite run of a pushdown
  system contains infinitely many configurations with effective
  stack-height one (see also~\cite{DEGM15}).
\end{remark}

Putting together Theorem~\ref{th:buechi} and Remark~\ref{lem:est1}, we obtain:

\begin{lemma}
  \label{lem:loop}
  There is a B\"uchi run in the $(\Ccf,\Dd)$-system iff there is one
  of the form
  \begin{equation*}
    (n \cs{\qin},t_\init, g_\init) 
    \xra{u} (M,t_1,g) \xra{v} (M,t_2, g) \xra{v} \ldots  
  \end{equation*}
   for some $n \in\Nat$,   $g \in G$, 
   $M \in (\Qn)^n$, $u,v\in(\Sc \cup \Sd)^*$ , where:
   \begin{itemize}
   \item $v$ ends by a letter from $\Sd$ and contains $\top$, and
\item all configurations
  $t_i \in \Qd \Gd^*$ of $\Dd$ have effective stack-height one, the same control
  state, and the same top stack symbol.
   \end{itemize}
\end{lemma}

\begin{proof}
Let $\rho= (n \cs{\qin}, t_\init, g_\init) = (M_0,t_0,g_0) \xra{a_1} 
  (M_1,t_1,g_1) \xra{a_2} \cdots$ be a B\"uchi run of the \CDsystem.
  By Remark~\ref{lem:est1} we can find an infinite set $I$ of
  indices such that for every $i\in I$:
  $t_i$ has effective stack-height one and $a_i\in\Sd$.
  For this we observe that if all configurations $t_i$ with $i$ greater
  than some number $i_0$
  have effective stack-height one, we can take all indices $i \ge i_0$ such
  that $a_i = \top$; otherwise  we take
  the set of indices such that $t_i$ has effective stack-height one, but
  $t_{i-1}$ does not -- then $a_i \in \Sd$.
  Since $\Qn$ is finite and $|M_i| = n$ for all $i$, the set 
  $\set{M_i \mid i \in \Nat}$ is finite. 
  By the pigeonhole principle, there exist $i, j \in I$ such that 
  $M_{i}=M_{j}$, $g_i=g_j$ and $t_i$, $t_j$ have
  effective stack-height one, the same state and the same top stack
  symbol. In addition, we ask that $\top$ is performed in the run
  from $(M_i,t_i,g_i)$ to $(M_j,t_j,g_j)$.

  We can then define a run of the desired form by repeating the part between 
  $i$ and $j$.
  We let $u = a_0 \cdots a_{i}$, $v = a_{i + 1} \cdots a_{j}$,
  $M = M_i = M_j$, $g = g_i = g_j$.
  To define the configurations $t_k$, we observe that the configurations 
  $t_j$ and $t_i$ can be represented as $t_i = q A \a$ and $t_j = q A \b \a$.
  We can then define $t'_k = q A \b^k \a$.
  We obtain a run of  the \CDsystem:
  $(n \cs{\qin}, t_\init, g_\init) \xra{u} (M,t'_1,g) \xra{v} 
  (M,t'_2, g) \xra{v} \ldots$
\end{proof}

\subsection{Capacities and supported loops}

The goal is a \PSPACE{} algorithm for the
existence of ultimately periodic runs in $(\Ccf,\Dd)$.
Since the reachability problem is decidable in
\PSPACE, we focus on loops. 
We follow the approach proposed in~\cite{LMW15} for the reachability
problem.  Adapting this approach to infinite runs is not
straightforward, and requires the new notion of $\omega$-support.

As in~\cite{LMW15}, we decompose a $(\Ccf,\Dd)$-system into a
finite-state system $\Ccfk$
representing the 
contribution of $\Ccf$, and a pushdown system $\Ddk$ representing the
contribution of~$\Dd$.


The idea underlying the decomposition is the following. Once a value
$g$ has been written by a contributor into the register, replicating
the contributor's run supplies arbitrary (but finitely) many
contributor writes of $g$. This is captured by introducing a new set
of actions $\S_\nu=\set{\nu(g) : g\in G}$ denoting first contributor
writes. In addition, each of $\Ccfk$ and $\Ddk$ have a component $K$
called \emph{capacity}, that stores the ``writing'' capacity that
contributors already provided. Formally, the set of control states of
$\Ddk$ is $\Pp(G)\times \Qd \times G$, and the initial state is
$(\es,\qid,g_\init)$. The input and the stack alphabets, $\Sd$ and
$\Gd$, are inherited from $\Dd$. So a configuration of $\Dd^\k$ has
the form $(K\incl G, \; t \in \Qd \Gd^*, \; g\in G)$. The transitions
of $\Dd^\k$ are:\label{def:Dk} 
\begin{align*}
    (K,t,g)\xra{\dw(h)}&(K,t',h)&&\text{if $t\xra{\dw(h)} t'$ in
                                  $\D$},\\
  (K,t,g)\xra{\dr(h)}& (K,t',h)&& \text{if $t\xra{\dr(h)} t'$ in
                                  $\D$ and $h\in K\cup\set{g}$},\\
  (K,t,g)\xra{\n(h)}& (K\cup\set{h},t,h)&& \text{if $h\not\in K$}\,.
\end{align*}

The finite transition system
$\Ccfk$ is defined similarly, it just follows in addition the
transitions of $\Ddk$. The set of states of $\Ccfk$ is $\Pp(G)\times \Qn\times G$, input alphabet
$\Sc$, and initial state $(\es,\qin,g_\init)$. The transition relation 
 $\d^\k$ is:
\begin{align*}
  (K,p,g)\xra{\dw(h)}& (K,p,h)\\  
  (K,p,g)\xra{\dr(h)}& (K,p,h)\\   
(K,p,g)\xra{\n(h)}& (K\cup\set{h},p,h) \\
  (K,p,g)\xra{\Cw(h)}& (K,p',h) \quad \text{if $p\xra{\Cw(h)}p'$ in $\d$ and
                                $h\in K$}\\
  (K,p,g)\xra{\Cr(h)}& (K,p',h) \quad \text{if $p\xra{\Cr(h)} p'$ in $\d$ and
                                $h\in K\cup\set{g}$}\,.
\end{align*}
Note that in both $\Ddk$ and $\Cck$ some additional reads
$\dr(h),\Cr(h)$ are possible when $h \in K$ -- these are called
\emph{capacity reads}.

\noindent\emph{Notation.} We write
$\S_{D,\n}$ for $\Sd\cup \S_\n$. Similarly for $\S_{C,\n}$ and
$\S_{C,D,\n}$. 
By $v|_\S$ we will denote the subword of $v$ obtained by erasing
the symbols not in~$\S$.
Note that the value of the register after executing a trace $v$, in
both $\Ccfk$ and $\Ddk$, is determined by the last action of $v$. We
denote by $\lst(v)$ the register value of the last action of $v$ (for
$v$ non-empty).


We now come back to examining when there exists an ultimately periodic
run of the $(\Ccf,\Dd)$-system, and focus on loops.  
Clearly, a loop in the $(\Ccf,\Dd)$-systems leads to a loop in $\Ddk$, 
but the converse is not true. 
To recover the equivalence, we introduce the notion of $\w$-supported traces. 
Informally, a loop $v$ of
$\Ddk$ will be called supported when (1) for each $\n(h)$ move in $v$
there is a trace of $\Ccfk$ witnessing the fact that a contributor run
can produce the required action $\Cw(h)$, and (2) all the witness
traces can be completed to loops in $\Ccfk$. 

\begin{definition}\label{def:osupp}
Consider a word
\begin{equation*}
v = v_1 \nu(h_1) \cdots v_m \nu(h_m) v_{m+1} \in \S^*_{D,\n},
\end{equation*}
where $v_1,\dots,v_{m+1}\in \Sd^*$, and $h_1,\dots,h_m$ are pairwise
different register values. 
We say that $v$ is \emph{$\w$-supported} from   
$p_1, \ldots, p_m \in \Qn$ if for every $1\leq i \leq m$ there is a word
$u^i\in(\S_{C,D,\n})^*$ of the form
\begin{equation*}
  u^i = u^i_1 \nu(h_1) \cdots u^i_i \nu(h_i) \boldsymbol{\Cw(h_i)} u^i_{i+1} \cdots
  u^i_m\nu(h_{m})u^i_{m+1}
\end{equation*}
such that: (i) $u^i|_{\S_{D,\n}} = v$, 
and (ii) $(\es, p_i, g) \xrightarrow{u^i}
(K,p_i,g)$ in $\Ccfk$, where $g=\lst(v)$.  
\end{definition}

Note that $K=\set{h_1,\ldots,h_m}$ in the above definition, and that
$u^i_j|_{\S_{D,\n}} = v_j$ holds for all~$j$. 

\begin{lemma}
  \label{lem:Clsupp}
  The \CDsystem has a B\"uchi run iff
  there is some reachable configuration $(M,q A \a,g)$ in the
    $(\Ccf,\Dd)$-system and a word $v\in\S^*_{D,\n}$ such that:
  \begin{enumerate}
    \item\label{item:b} $\Ddk$ has a run of the form
    $(\es, q A, g)\xrightarrow{v} (K, q A \a', g)$, and $\top$ appears
    in~$v$.
  \item\label{item:c} $v$ is \Clsupp from some $p_1, \ldots, p_m$ such
      that $[p_1,\dots,p_m]\leq M$.
  \end{enumerate}
\end{lemma}

Observe that by Definition~\ref{def:osupp}, we have $m\le |G|$ in
Lemma~\ref{lem:Clsupp}.

\begin{proof}[Proof sketch]
  We outline the proof, focussing on the
  right to left direction. We construct an infinite run of the
  $(\Ccf,\Dd)$-system starting in $(M,q A \a,g)$ by shuffling 
  in a suitable way (infinitely many copies of) the run $v$ of $\Ddk$ and a number of 
  copies of the runs of $\Ccfk$ supporting $v$.

  Consider one of the $\n(h)$ occuring in $v$. Since $v$ is $\w$-supported,
  there is a run $\rho: p \xra{u_1} p_1 \xra{\Cw(h)} p_2 \xra{u_2} p$ of
  $\Ccfk$ that can be executed with $v$ to produce the first
  occurrence of $\Cw(h)$, in place of $\n(h)$.
  To execute $v$ once, we can replicate the 
  initial part $p \xra{u_1} p_1$  as many times as needed, and
  simulate each capacity read $\Dr(h)$ or $\Cr(h)$ occuring later, in $v$
  or one of its supporting runs, by letting a different copy take the 
  transition $p_1 \xra{\Cw(h)} p_2$ to enable the read.
  At the end of the first simulation of $v$, the main contributor executing 
  $\rho$ is back in state $p$, and the others are still in state $p_2$.
  During the simulation of the second loop, we use a
  second set of copies of the main contributor to produce in the
  same way all required write operations $\Cw(h)$.  Meanwhile, the first
  set of contributors can be brought back 
  to state $p$: as soon as the main contributor reaches again state
  $p_2$, the first set of copies resumes the run $p_2 \xra{u_2} p$
  by following the run of the  main contributor.
\end{proof}

\subsection{Final step} 

As in the case of reachability, we show that
$\Ddk$ can be replaced by a finite-state system representing its downward
closure, since adding some transitions of the leader does not affect the
support of contributors. This finite-state system will be 
synchronized with the contributor automata witnessing support,
yielding the \PSPACE{} algorithm.



\begin{lemma}
  \label{lem:subwords}
  Let 
  $v=v_1 \nu(h_1) \cdots v_{m+1}$be \Clsupp from $p_1,
  \ldots, p_m$, and let $v_j \sqsubseteq \bar{v}_j$ for every
  $j$. Assume that 
  $\bar{v}=\bar{v}_1 \nu(h_1) \cdots \bar{v}_{m+1}$
  satisfies $\lst(v)=\lst(\bar{v})$. Then $\bar{v}$
  is also \Clsupp from $p_1, \ldots, p_m$.
\end{lemma}

\begin{proof}
  We need to lift all traces $u^i$ of $\Ccfk$ that
  witness that $v$ is \Clsupp, to traces $\bar{u}^i$ that $\w$-support
  $\bar{v}$. We assume that the traces $u^i$ satisfy all assumptions
  of Definition~\ref{def:osupp}, 
  and write $g = \lst(v) = \lst(\bar v) = \lst(u^i)$.

  Let $v_i=a_{i,1} \cdots a_{i,n_i}$ and $\bar{v}_i=x_{i,0} a_{i,1}
  x_{i,1} \cdots a_{i,n_i} x_{i,n_i}$, for some $x_{i,j} \in
  \Sd^*$. We obtain $\bar{u}^i$ from $u^i$ by substituting each $a_{i,j}$ by
  $x_{i,j-1}a_{i,j}$, each $\n(h_l)$ by $x_{l,n_l}\n(h_l)$ and 
  adding $x_{m+1,n_{m+1}}$ at the end.  By construction we have 
  $\bar{u}^i|_{\S_{D,\n}}=\bar{v}$,
  and $\lst(\bar u^i) = g$, since $\bar u^i$ ends either with the same action 
  as $x_{m+1,n_{m+1}}$ (i.e., as $\bar v$), or as $u^i$.
  Observe also that $(\es, p_i, g) \xrightarrow{\bar{u}^i} (K,p_i,g)$ still 
  holds in $\Ccfk$: actions of $\Dd$ can only modify the register component
  in $\Ccfk$, and this is the same after reading $a_{i,j}$ or $x_{i,j-1}a_{i,j}$.
\end{proof}

Combining known results for the reachability problem in \CDsystems with Lemmas \ref{lem:Clsupp} and 
\ref{lem:subwords}, we obtain a polynomial space algorithm  for the
repeated reachability problem:

\begin{theorem}\label{thm:liveness1}
  The repeated reachability problem for \CDsystems is \PSPACE-complete
  when $\Cc$ and $\Dd$   range over pushdown systems.
\end{theorem}

\begin{proof}
  The lower bound follows from~\cite{EGM16}.

  For the upper bound we introduce one more shorthand. Let $h_1,
  \ldots, h_m$ be a sequence of values from $G$. 
  A $(h_1,\dots,h_m)$-word is a word of the form $v_1\n(h_1)\cdots
  v_m\n(h_m)v_{m+1}$ with no occurrence of $\n$ in $v_1\cdots v_{m+1}$.

  Our \PSPACE\ algorithm consists of three steps:
  \begin{enumerate}
  \item Guess a sequence $h_1, \ldots, h_m$ of pairwise distinct values
    from $G$, states $p_1,\ldots,p_m$ of $\Ccf$, control state $q \in
    Q$ and stack symbol $A \in \Gd$ for $\Dd$, and a value $g\in G$.     

  \item Check that there exists a configuration $(M,qA\a,g)$ 
    satisfying $M \ge [p_1, \ldots, p_m]$ that is reachable in the 
    $(\Ccf,\Dd)$-system.
  \item Check that there exists a $(h_1,\dots,h_m)$-word $v \in \S_{D,\n}^*$ 
    with $\lst(v) = g$ such that:
    \begin{enumerate}
    \item $v$ is \Clsupp from $p_1, \ldots, p_m$,
    \item $\Ddk$ has a run of the form
      $(\es,qA,g) \xra{\bar v} (K,qA\a',g)$ for some
      $(h_1,\dots,h_m)$-word $\bar v$ such that 
      $v \sleq \bar v$, $\lst(\bar v) = \lst(v) = g$, and $\top$
      occurs in $\bar v$.
    \end{enumerate}
  \end{enumerate}
  By Lemmas \ref{lem:Clsupp} and \ref{lem:subwords}, this algorithm
  returns ``yes'' iff the system has a B\"uchi run.

  \smallskip
  Step~1 can be done in polynomial space since we have $m \le |G|$.
  \smallskip

  Step~2 reduces to an instance of the reachability problem in the \CDsystem.
  First, by Lemma~\ref{lem:Ccf}, the question is
  to decide if there exists a configuration $(M,qA\a,g)$ in the \CDsystem that
  is reachable and such that $M \ge [p_1\a_1, \ldots, p_m\a_n]$ for some
  $\a_1, \ldots, \a_n \in \G^*$.
  We modify $\Cc$ and $\Dd$ so that after such a configuration has been
  reached (and only in that case), the leader can do a new action $\top_0$.
  The idea is to add transitions to $\Cc$ so that a contributor in a state
  of the form $p_i\a$ can write ``$i$'' to the register (once), and add
  transitions $qA \xra{\dr(1)} \ldots \xra{\dr(m)} \xra{\top_0}$
  to $\Dd$.
  Since the reachability problem is in \PSPACE\ when $\Dd$ and $\Cc$
  are pushdown systems~\cite{EGM16,LMW15}, step 2 is in
  \PSPACE.

  \smallskip

  Step~3 requires to construct some auxiliary automata.
  For $i = 1, \ldots, m$, let $\Aa_i$ be a finite automaton
  accepting the projection over $\S_{D,\n}$ of the words $u \in
  \S_{D,C,\n}^*$ 
  of the form
  \begin{equation}\label{eq:runCck}
  u = u_1 \nu(h_1) \cdots u_i \nu(h_i) \boldsymbol{\Cw(h_i)}  u_{i+1} \cdots u_m\nu(h_{m})u_{m+1}
  \end{equation}
  and such that $(\es, p_i, g) \xrightarrow{u} (K,p_i,g)$ is a trace in $\Ccfk$.
  Since $\Ccfk$ can be computed in \PSPACE{}, so does $\Aa_i$.

  Consider the language:
  \begin{align*}
    L=\{ \bar v\in \S^*_{\Dd,\n}\ :\quad
    & \text{$\bar v$ contains $\top$,} \;\; \lst(\bar v)=g, \;
    \text{ $\bar v$ is a $(h_1,\dots,h_m)$-word,}\\
    & \text{$(\es,qA,g) \xra{\bar v} (K,qA\a,g)$ in $\Ddk$, for some
      $\a$}\}\,.
  \end{align*}
  A pushdown automaton recognizing $L$ can be obtained by slightly
  modifying $\Ddk$.  Since $L$ can be recognized by a pushdown
  automaton of polynomial size, it is possible to compute on-the-fly
  in \PSPACE\ a finite automaton $\Aa$ (of exponential size)
  accepting all the $(h_1,\dots,h_m)$-words in the downward closure
  of~$L$, see~\cite{cou91beatcs}.
  
  Step~3 then consists in checking that the intersection of the automata
  $\Aa$, $\Aa_1,\ldots,\Aa_m$ is non-empty, which can be done in \PSPACE{}.
\end{proof}

\begin{remark}
  In the case where there is no leader, $\Ddk$ becomes an automaton accepting
  all sequences $\n(h_1) \cdots \n(h_m)$ such that $h_1, \ldots, h_m$ are
  pairwise distinct.
  To check that such a sequence is $\omega$-supported, we can test separately
  for each $\n(h_i)$ if there is a contributor run that can produce
  $\Cw(h_i)$ -- instead 
  of having to take the product of $m$ automata in order to
  synchronize with $\Ddk$, as in the proof of
  Theorem~\ref{thm:liveness1}.
  For that reason, we can test directly for each $\n(h_i)$ if the pushdown 
  automaton $\Cck$ (rather than $\Ccfk$) admits a run as in 
  Equation~\ref{eq:runCck}.
  This leads to an algorithm in \NP\ instead of \PSPACE.
\end{remark}


\section{Max-safe problem}
\label{sec:universal}

We show in this section that the max-safe problem is \NP-complete when
$\Cc$ ranges over finite-state systems and $\Dd$ ranges over 
pushdown systems, and \NEXPTIME-complete when both $\Cc$ and $\Dd$ range 
over pushdown systems (Theorem~\ref{th:safety}).

We start by introducing a set semantics of \CDsystems, that replaces
multisets by sets.  This semantics is suitable for the  reachability
and max-safe problems (but not for liveness). We show
that the max-safe problem is \NP-complete when contributors are
finite-state.  Then we consider the case of contributors given by a
pushdown automaton. As for liveness we can reduce this case to the case when
contributors are finite-state.  This gives  a \NEXPTIME\ algorithm.

\subsection{Set semantics}

As a first step we will introduce the set semantics of \CDsystems that is equivalent to the
multiset semantics of Section~\ref{sec:prelim} when only finite traces are considered.
The idea is that since the number of contributors is arbitrary, we can always 
add contributors that copy all the actions of a given contributor. So once a
state of $\Cc$ is reached, we can assume that we have arbitrarily many copies 
of $\Cc$ in that state.
In consequence, we can replace multisets by sets.
A very similar semantics has already been used
in~\cite{LMW15,DEGM15}.
Here we need to be a bit finer in order to handle deadlocks.

Consider a \CDsystem with the notations as in~Eq.~\eqref{eq1} on page~\pageref{eq1}:
\[
  \Cc=\struct{S,\d,s_\init}\qquad
  \Dd=\struct{T,\D, t_\init}\ .
\]

Instead of multisets $M \in \Nat^S$,
we use sets $B \subseteq S$. As for multisets we lift the transitions
from elements to sets of elements:
\begin{align*}
  B \xra{a} B'\text { in $\d$}&\quad\text{if $s \xra{a} s'$ in
                                $\delta$, and $B'$  is
either $B \cup \set{s'}$ or $(B \cup\set{s'}) \setminus \set{s}$}\\
  &\quad\text{for some $s\in B$.}
\end{align*}
The intuition is that $B \xra{a} B \cup \set{s'}$ represents the case where
\emph{some} contributors in state $s$ take the transition,
and $B \xra{a} (B \cup \set{s'}) \setminus \set{s}$ corresponds to the case 
where \emph{all} contributors in state $s$ take the transition.
The transitions in the \emph{set semantics} are essentially the same
as for the multiset case:
\begin{align*}
  (B,t,g)\xra{\dw(h)}& (B,t',h)&& \text{if $t\xra{\dw(h)} t'$ in
                                  $\D$} \\
  (B,t,g)\xra{\dr(h)}& (B,t',h)&& \text{if $t\xra{\dr(h)} t'$ in
                                  $\D$ and $h=g$}\\
  (B,t,g)\xra{\Cw(h)}& (B',t,h)&& \text{if $B\xra{\Cw(h)} B'$ in $\d$} \\
  (B,t,g)\xra{\Cr(h)}& (B',t,h)&& \text{if $B\xra{\Cr(h)} B'$ in $\d$ and $h=g$}
\end{align*}

\begin{lemma}\label{lemma:multi-set}
  \begin{enumerate}
  \item If $(M_0,t_0,g_0) \xra{a_1} \ldots \xra{a_n} (M_n,t_n,g_n)$ in the
    multiset semantics, and $B_j$ is the support of $M_j$, for
    every $j=0,\dots,n$, then
    $(B_0,t_0,g_0) \xra{a_1} \ldots \xra{a_n} (B_n,t_n,g_n)$ in the set
    semantics.
  \item If $(B_0,t_0,g_0) \xra{a_1} \ldots \xra{a_n} (B_n,t_n,g_n)$ in the
    set semantics, then there exist multisets $M_0, \ldots, M_n$ such that
    $M_j$ has support $B_j$, and for some $i_j >0$,
    \[(M_0,t_0,g_0) \xra{(a_1)^{i_1}} (M_1,t_1,g_1) \xra{(a_2)^{i_2}} 
    \ldots \xra{(a_n)^{i_n}} (M_n,t_n,g_n)\] in the multiset semantics.
  \end{enumerate}
\end{lemma}

\begin{proof}
  The first part of the lemma follows directly from the definitions.
  For the second part, writing $B_0 = \set{s_1,\ldots,s_n}$,
  we let $M_0 = 2^n [s_1, \ldots, s_m]$.
  We then simulate a step of the set semantics by letting either the leader,
  all the copies of $\Cc$, or half the copies of $\Cc$ take the transition
  in the \CDsystem.
\end{proof}

\begin{remark}
  The set semantics is a variant of the \emph{accumulator semantics}
  used in~\cite{LMW15}, in which only transitions of the form $B
  \xra{a} B \cup \set{s'}$ (but not $B \xra{a} (B \cup \set{s'})
  \setminus \set{s}$) were used. 
  The accumulator semantics is sufficient for the reachability
  problem, and
  has the nice property that the $B$-part is monotonic (hence the name 
  \emph{accumulator} semantics). 
  So for instance, in the case of finite-state 
  processes, it leads to a very simple \NP{} algorithm for the reachability 
  problem \cite{LMW15}.
  However, the accumulator semantics does not satisfy the first item
  of Lemma~\ref{lemma:multi-set}, and is thus not precise enough for
  properties that refer to~\emph{maximal} runs.
\end{remark}

\begin{corollary}\label{cor:set-blocking}
  Fix a \CDsystem.
  In the multiset semantics  the system has a finite maximal safe run
  ending in a configuration $(M,t,g)$ iff in the set semantics the
  system has a finite maximal safe run ending in the configuration $(B,t,g)$ 
  with $B$ being the support of $M$.
\end{corollary}

\subsection{Finite-state contributors}

First we will assume that $\Cc$ is  a finite-state transition system.

\begin{lemma}
  \label{lem:NPhard}
  The max-safe problem is \NP-hard when $\Cc$ and $\Dd$
  are both finite-state.
\end{lemma}

\begin{proof}
  We reduce 3-SAT to the max-safe problem.
  Given a formula $\varphi = c_1 \land \ldots \land c_m$, with $c_j$
  clauses of length 3  over variables 
  $x_1, \ldots, x_n$, we construct finite-state processes $\Cc$ and $\Dd$ 
  such that $\varphi$ is satisfiable iff there exists 
  a maximal  run in the \CDsystem that contains no occurrence of a
  fixed action $\top$ of $\Dd$.
  
  The leader $\Dd$ will guess the values of the variables, and the contributors
  will check if all clauses are satisfied.
  So, the leader starts by successively writing 
  $x_1 = b_1, \ldots, x_n = b_n$ in the register, where each $b_i$ is a
  guessed truth value.
  Meanwhile, each contributor chooses a clause $c_j$, reads the values
  of the variables appearing in $c_j$ as the leader writes them, and
  checks whether $c_j$ is satisfied. If it is, he writes ``$c_j$'' in the 
  register.
  After having guessed values for the all the variables, the  leader
  must successively read ``$c_1$'', \ldots, ``$c_m$''.
  If she manages to do this then she knows that  all clauses are satisfied.
  She enters then some state $q$  with no outgoing transitions.
  For every other state $q' \neq q$ of the leader we add a transition
  $q'\xra{\top}q'$.

  Suppose $\varphi$ is satisfiable. Take  a run as described above,
  starting with $n$ contributors, and where the leader chooses a
  valuation that satisfies $\varphi$.
  The leader ends in state $q$, from which she has no available transition.
  Similarly, each contributor stops after writing one of the
  ``$c_i$'', or blocks because he missed reading some variable. 
  So the run is maximal, and does not contain~$\top$.

  For the other direction, observe that all safe runs must be finite
  and should end with $\Dd$ in the state $q$ because all other states
  have a transition on  $\top$. 
  Such a run defines a valuation of the variables that satisfies
  $\varphi$.
\end{proof}

To decide the max-safe problem, we check separately for
the existence of an infinite, or finite and maximal, run without
occurrences of $\top$ (a \emph{safe run}).
The case of infinite safe runs can be reduced to the repeated
reachability problem: by Theorem~\ref{th:transP} we can
construct from $(\Cc,\Dd)$ an equivalent \CDtsystem in which all infinite runs 
contains infinitely many writes from the leader. To decide if this system 
admits an infinite run, we can then test for each possible value $g$ of the 
register if there is a run with infinitely many occurrences of $\Dw(g)$.
Since the repeated reachability problem is in \NP{} for finite-state
contributors~\cite{DEGM15} we obtain:

\begin{lemma}\label{l:inf}
  When $\Cc$ ranges over finite-state systems and $\Dd$ over pushdown systems, 
  deciding whether a \CDsystem has an infinite safe run is in \NP.
\end{lemma}

It remains to give an algorithm for the existence of a \emph{finite} maximal
safe run.
By Corollary~\ref{cor:set-blocking} we can use the set
semantics.
From now on we will also need to exploit the fact that $\Dd$ is a
pushdown system. 
Recall that the states of $\Dd$ are of the form $q\a$ where $q$ is the
state of the pushdown automaton defining $\Dd$ and $\a$ represents the
stack. 
The question is to decide if there is a configuration $(B,q\a,g)$ 
from which there is no outgoing transition in the \CDsystem, and such that
$(B,q\a,g)$ is reachable without using $\top$ actions.  Note that
we can say whether $(B,q\a,g)$ has no outgoing transition by looking
only at $B,q,g$ and the top symbol of $\a$.
Our algorithm will consists in guessing $B,q,g$ and some $A \in\G_D$, 
and checking reachability.
First, we show that it is sufficient to look for traces where the
number of changes to the first component of configurations is bounded.
The idea is that we can always assume that in a run of the \CDsystem, a
state $s$ is added at most once to the current set of contributor states $B$.
This simply means that a state is removed from $B$ only if it will never be
used again in the run.

\begin{lemma}
  \label{lem:Bseq}
  Let $\rho = (B_0, t_0, g_0) \xra{a_1} (B_1, t_1, g_1) \xra{a_2} \ldots 
  \xra{a_n} (B_n,t_n,g_n)$ be a run of the \CDsystem{}. There exists a run 
  $\rho' = (B'_0, t_0, g_0) \xra{a_1} (B'_1, t_1, g_1) \xra{a_2} \ldots 
  \xra{a_n} (B'_n,t_n,g_n)$ such that $B_0 = B'_0$, $B_n = B'_n$, 
  and for all $s \in S$ and $0 \le i < n$,
  if $s \in B'_i$ and $s \notin B'_{i+1}$, then for all $j > i$, $s
  \notin B'_j$.
\end{lemma}

\begin{proof}
  We define $B'_i$ by induction on $i$: $B'_0 = B_0$, and for $i > 1$,
  \begin{itemize}
  \item if $B_{i+1} = B_i$, then $B'_{i+1} = B'_i$.
  \item if $B_{i+1} = B_i \cup \{s\}$, then $B'_{i+1} = B'_i \cup \{s\}$.
  \item if $B_{i+1} = (B_i \setminus \{s\}) \cup \{s'\}$ and $s \notin B_j$
    for all $j > i$, then $B'_{i+1} = (B'_i \setminus \{s\}) \cup \{s'\}$.
    If $s \in B_j$ for some $j > i$, then $B'_{i+1} = B'_i \cup \{s'\}$.
  \end{itemize}
  Clearly, $\rho'$ is a run of the \CDsystem{}.
  Moreover, for all $i$, $B_i \subseteq B'_i \subseteq \bigcup_{j = i}^n B_j$.
  So in particular, $B_n = B'_n$.
\end{proof}

\begin{corollary}
Every finite run $\r$ of the \CDsystem in the set semantics can be written as $\r=\r_0 \cdots
\r_k$ with $k \le 2|\Qcf|$, where in each $\r_j$, all states have
the same first component.
\end{corollary}

\begin{proof}
  We take a run of the form described in Lemma~\ref{lem:Bseq}.
  Let $i_0 = 0$, and $i_1 < \cdots < i_k$ be the indices such that 
  $B_i \neq B_{i-1}$.
  Consider the sequence $B_{i_0}, B_{i_1}, \ldots, B_{i_k}$.
  There are states $s_1, \ldots, s_k \in S$ such that for all $0 \le j < k$,
  $B_{i_{j+1}} = B_{i_j} \cup \set{s_j}$ or 
  $B_{i_{j+1}} = (B_{i_j} \cup \set{s}) \setminus \set{s_j}$ for some $s$.
  Moreover, each $s\in S$ is added at most once, and removed at most once
  from some $B_i$, which means that there are at most two distinct indices 
  $j$ such that $s = s_j$. Hence $k \le 2|\Qcf|$. 
\end{proof}

\begin{lemma}
  \label{lem:boundedruns}
For every finite run $\r$ of the \CDsystem in the set semantics, there exists
a run $\r'$ with same label and end configuration that can be written 
as $\r'=\r_0 \cdots \r_k$ with $k \le 2|\Qcf|$, where in each $\r_j$, all 
states have the same first component.
\end{lemma}

\begin{proof}
  We take a run of the form described in Lemma~\ref{lem:Bseq}.
  Let $i_0 = 0$, and $i_1 < \cdots < i_k$ be the indices such that 
  $B_i \neq B_{i-1}$.
  Consider the sequence $B_{i_0}, B_{i_1}, \ldots, B_{i_k}$.
  There are states $s_1, \ldots, s_k \in S$ such that for all $0 \le j < k$,
  $B_{i_{j+1}} = B_{i_j} \cup \set{s_j}$ or 
  $B_{i_{j+1}} = (B_{i_j} \cup \set{s}) \setminus \set{s_j}$ for some $s$.
  Moreover, each $s\in S$ is added at most once, and removed at most once
  from some $B_i$, which means that there are at most two distinct indices 
  $j$ such that $s = s_j$. Hence $k \le 2|\Qcf|$. 
\end{proof}

\begin{lemma}
  \label{lem:GlobalReachNP}
  The following problem belongs to \NP:

\noindent
  \textbf{Input:} 
  finite-state system $\Cc$, pushdown system $\Dd$,
  $B \subseteq S$, $q \in Q$, $A \in\Gd$.

\noindent
  \textbf{Question:} Does the \CDsystem admit a run from
  $(\set{s_\init},\qid \Aid, g_\init)$ to $(B,qA\a,g)$ for some $\a \in\Gd^*$?
\end{lemma}

\begin{proof}
  The set semantics of the \CDsystem can be described by a pushdown automaton
  $\Aa$ with set of control states $2^\Qcf \times \Qd \times G$, input alphabet
  $\Sc \cup \Sd$, and stack alphabet $\Gd$.
  We guess a sequence $\{\qicf\} = B_0, B_1, \ldots, B_k = B$ where
  $k \le 2|\Qcf|$, and construct the restriction of the pushdown automaton $\Aa$
  to runs where the first component of the state takes values
  according to $B_0, B_1, \ldots, B_k$. 
  This new pushdown automaton is of polynomial size, and we can check whether
  it has a reachable configuration $(B,qA\a,g)$ in polynomial time~\cite{BEM97}.
\end{proof}

By Lemmas~\ref{l:inf} and \ref{lem:GlobalReachNP}, the max-safe
problem is thus in \NP{} and we obtain:

\begin{theorem}\label{th:safety-finite}
  The max-safe problem is \NP-complete when $\Cc$ ranges over
  finite-state systems and $\Dd$ ranges over   pushdown systems.
\end{theorem}

\subsection{Pushdown contributors}

We now return to the case where
both $\Cc$ and $\Dd$ are given by pushdown systems.

\begin{lemma}\label{l:safety-hardness}
 The max-safe problem is \NEXPTIME-hard when $\Cc$ and $\Dd$ range over
  pushdown systems.
\end{lemma}

\begin{proof}
  We reduce the following tiling problem to the
  max-safe problem:

\medskip

\noindent
  \textbf{Input:} A finite set of tiles $\Sigma$, 
  horizontal and vertical compatibility relations
  $H,V\subseteq\Sigma^2$, and initial row $x \in \S^n$.

\noindent
  \textbf{Question:} is there a tiling of the $2^n \times 2^n$ square
  respecting the compatibility relations and containing the initial
  row in the left corner?

 A tiling is a function
  $t : \set{1,\ldots, 2^n}^2 \to \S$ such that $(t(i,j),t(i,j+1)) \in H$
  and $(t(i,j),t(i+1,j)) \in V$ for all $i,j$, 
  and $t(1,1)t(1,2) \cdots t(1,n) = x$.

  \smallskip

  The idea of the reduction is that the system will have a maximal 
  run without $\top$ if and only if the leader guesses a tiling respecting the horizontal 
  compatibility, and the contributors check that the vertical 
  compatibility is respected as well.
  
  The leader will write down the tiling from left to right and from
  top to bottom, starting with the initial row. 
  The sequence of values taken by the register on a (good) run will have the form
  \begin{equation*}
    A_{1,1}, \overline{A_{1,1}}, A_{1,2}, \overline{A_{1,2}}, \ldots,
    A_{1,2^n}, \overline{A_{1,2^n}},\ldots, A_{2^n, 2^n} \,
    \overline{A_{2^n, 2^n}}\; (\$ \overline{\$})^{2^n}\diamond\, .
  \end{equation*}
  The $A_{i,j}$ are guessed and written by the leader, and the
  $\overline{A_{i,j}}$ are written by contributors. Letters
  $\overline{A_{i,j}}$ have two purposes: they ensure that at least
  one contributor has read the preceding letter, and prevent a
  contributor to read the same letter twice.
  For technical reasons, this sequence is followed by a sequence 
  $(\$ \overline{\$})^{2^n}\diamond$ of writes from the leader (with
  $\$,\diamond \notin \Sigma$), 
  and we will consider that $(A,\$) \in V$ for all $A \in \Sigma$.
  
   The leader uses her stack to count the number $i$ of rows
  (using the lower part of the stack), and the number $j$ of tiles on each row
  (using the upper part of the stack). So, she repeats the following, up to
  reaching the values $i = 2^n, j = 2^n$:
  \begin{itemize}
  \item guess a tile $A$ compatible with the one on its left (if $j \neq 1$),
    and write $A$ on the register,
  \item wait for an acknowledgment $\overline{A}$ from one of the contributors,
  \item increment $j$,
  \item if $j >2^n$, increment $i$ and set $j=1$.
  \end{itemize}
  Finally, she repeats $2^n$ times the actions $\dw(\$)$,
  $\dw(\overline{\$})$, then finishes by writing  $w(\diamond)$ and
  going to state $q_f$. 

  Each contributor is supposed to read the entire sequence of values written in the register.
  He alternates between reading values of the form $A$ and $\overline{A}$, 
  which ensures that no value is read more than one
  time. At the same time, he uses his stack to count the number of
  writes 
  $\dw(A)$ ($A \in \Sigma \cup \{\$\}$) of the leader, up to 
  $(2^{2n}+2^n)$, so that he can check that no value was missed. 
  This operation will in fact be divided between counting up to
  $2^{2n}$, and counting up to $2^{n}$, as described below.

  Every contributor decides non-deterministically to 
  check vertical compatibility at some point. He chooses the current
  tile $A \not=\$$, and
  needs to check
  that the tile located below it (that is, occurring $2^n$ tiles later in the
  sequence of values written by the leader) is compatible with it.
  This is done as follows: after reading $A\not=\$$, the contributor writes
  $\overline{A}$ on  
  the register (rather than waiting for another contributor to do so), and 
  remembers the value. 
  He interrupts his current counting, and starts counting anew on the top
  of the stack, up to $2^n$. Upon reaching $2^n$, he stores the value
  $A'$ of the register, for later check. Then he resumes the first counting while reading the remaining of the 
  sequence, up to $2^{2n}$. At any moment, the contributor can read
  $\diamond$. If he reads $\diamond$ and either $(A,A') \notin V$ or the counting up to $2^n$ failed (i.e.,
 his stack is not empty), then he writes
  $\# \notin G$ and stops; otherwise he simply stops. 
 In state $q_f$, the leader may read any value $g \neq \diamond$,
 and she then do $\top$: $q_f \act{\Dr(g)} \act{\top}$.
 From every other state $q \not= q_f$, the leader can do $\top$, too.

  If there is a tiling of the $2^n \times 2^n$ square, then we obtain
  a maximal run with  $2^{2n}$ 
  contributors and  without any occurrence of $\top$, by letting the
  leader write the sequence 
  of register values corresponding to this tiling, and having each
  contributor  perform one of the $2^{2n}$ vertical compatibility
  checks. If each contributor reads every value produced by the
  leader, his stack will be empty upon reading $\diamond$, so he simply
  stops and no $\top$ will be generated.

  Conversely, we show that in any maximal run without $\top$, the sequence of tiles 
  guessed by the leader defines a correct tiling of the $2^n \times 2^n$ square.
  First, in any such run the leader needs to reach state $q_f$: 
  if she gets no acknowledgment on some $A \in \S$ then she would do
  $\top$, which is impossible by assumption. So she guesses 
  a sequence $A_{1,1}, \ldots, A_{2^n,2^n}$ with 
  $(A_{i,j},A_{i,j+1}) \in H$ for all $j < 2^n$, gets an acknowledgment 
  $\bar A_{i,j}$ for each $A_{i,j}$, and finally writes $(\$
  \overline{\$})^{2^n}\diamond$. Moreover, $\diamond$ is the final value of the
  register: if it were overwritten by a contributor, the leader could generate
  $\top$. So every contributor that has chosen some $A_{i,j}
  \in\S$ will ultimately read $\diamond$. Since he cannot do $\Cw(\#)$, his
  stack must be empty at that point and he must have successfully checked that
  $(A_{i,j},A_{i+1,j}) \in V$. 
\end{proof}

As with finite-state contributors, to solve the max-safe problem, we look
separately for an infinite safe run, or a finite maximal safe run.
The case of infinite runs can again be reduced to the repeated reachability
problem, using Theorem~\ref{th:transP}.

\begin{lemma}
  When $\Cc$ and $\Dd$ range over  pushdown systems, deciding whether 
  a \CDsystem has an infinite safe run is in \PSPACE.
\end{lemma}

To decide the existence of a \emph{finite} maximal safe run, we reduce
the problem to the case of finite-state contributors,
using Lemma~\ref{lem:Ccf}.

\begin{lemma}\label{lem:NEXPup}
  When $\Cc$ and $\Dd$ range over  pushdown systems, deciding whether 
  a \CDsystem has a finite maximal safe run is in \NEXPTIME.
\end{lemma}

\begin{proof}
We define the \emph{top} of a configuration of the \CDsystem as follows:
\begin{equation*}
  \ttop(\set{p_1A_1\a_1,\ldots,
    p_nA_n\a_n},qA\a,g)=(\set{p_1A_1,\ldots, p_nA_n},qA,g) 
\end{equation*}
Observe that to determine if a configuration in the set semantics is a
deadlock or not it suffices to look at its top. Moreover, by 
Lemma~\ref{lemma:multi-set}, deadlocks occur in the multiset semantics iff 
they occur in the set semantics.

The algorithm to decide the existence of a maximal safe run is as follows:
guess a configuration top that corresponds to deadlocks in the \CDsystem,
and check if it is reachable after removing all $\top$-transitions from
the \CDsystem.
By Lemma~\ref{lem:Ccf}, this amounts to deciding if it is reachable in
the $(\Ccf,\Dd)$-system. Applying Lemma~\ref{lem:GlobalReachNP} to the
$(\Ccf,\Dd)$-system, which is of exponential size, this can be done in
\NEXPTIME.
\end{proof}

\begin{theorem}\label{th:safety-pushdown}
  The max-safe problem is
  \NEXPTIME-complete when $\Cc$ and $\Dd$ range over pushdown
  systems.
\end{theorem}

\subsection{Universal reachability}

Recall that the \emph{universal reachability} problem asks if all maximal runs of
a given \CDsystem, so for every number of contributors, contain some occurrence of a special action
$\top$.
 Correctness problems for parametrized distributed algorithms
can be rephrased as instances of universal reachability: we want to
know whether for an arbitrary number of participants, and for every
run of the algorithm, the outcome is correct.  Correctness of the
outcome is expressed here by the leader executing the action $\top$.

\begin{remark}
  A natural variant of the universal reachability problem would be the
  following: is there some bound $N$ such that for all $n \ge N$, all
  maximal runs with $n$ contributors contain an occurrence of $\top$?
  A bit surprisingly, this formulation is equivalent to the universal
  reachability problem: if there were some maximal run with $n<N$
  contributors without $\top$, then we could add arbitrary many
  contributors doing the same actions as one original contributor,
  thus obtaining a maximal run with $N$ contributors  and without
  $\top$, contradiction.
\end{remark}

Since the max-safe problem is the complement of universal reachability, we obtain from Theorems~\ref{th:safety-finite} and \ref{th:safety-pushdown}:

\begin{corollary}
  The universal reachability problem is \coNP-complete for \CDsystems
  where $\Cc$ is finite-state and $\Dd$ is a pushdown system. It is
  \coNEXPTIME-complete when both $\Cc$ and $\Dd$ are pushdown systems. 
\end{corollary}


\section{Regular $\Cc$-expanding properties}
\label{sec:general}

In this section, we prove our general result stated in
Theorem~\ref{th:main}. 


The proof of Theorem~\ref{th:main} is (again) divided into two cases: 
one for finite and the other for infinite traces.
For finite maximal traces we use the results about the max-safe problem,
and for infinite traces we give a reduction to repeated reachability.

 \begin{lemma}\label{l:property-finite}
    It is \NEXPTIME-complete to decide whether a given pushdown
    \CDsystem has a \emph{finite} maximal trace satisfying some 
    $\Cc$-expanding property $\Pp$
    given by a finite automaton or an LTL formula.
  \end{lemma}

  \begin{proof} By Theorem~\ref{th:transP} we can assume that we
deal with a property $\Pp_D$ referring only to actions of $\Dd$. 
If $\Pp_D$ is given by an LTL formula, we start by constructing an equivalent finite
automaton of exponential size.
 By taking the product of $\Dd$ with this automaton, we can
assume that $\Dd$ has a distinguished set of final (control) states
such that a finite run of the \CDsystem satisfies $\Pp_D$ iff $\Dd$ ends
in a final state.

The result then follows using Lemma~\ref{lem:Ccf}, together
with Lemma~\ref{lem:GlobalReachNP}. 
Recall that in order to decide if a finite run is maximal it is enough
to look at the top of its last configuration. 
Lemma~\ref{lem:Ccf} then tell us that there exists a maximal finite run in the
\CDsystem with $\Dd$ ending in a final state iff there exists such a
run in the $(\Ccf,\Dd)$-system; and by Lemma~\ref{lem:GlobalReachNP} 
this can be decided in \NP{} in the
size of $(\Ccf,\Dd)$, so overall in \NEXPTIME. 
The matching \NEXPTIME-hardness lower bound follows from
the proof of Lemma~\ref{l:safety-hardness}, 
as the \CDsystem constructed there has no 
infinite safe trace, and the max-safe problem restricted to finite traces is a 
special case of our problem.
  \end{proof}

  The case of infinite runs turns out to be easier complexity-wise:
 \PSPACE{} if the property is given by an automaton, and \EXPTIME{} if it
  is given by an LTL formula.

  \begin{lemma}\label{l:property-infinite}
    It is \PSPACE-complete to decide whether a given pushdown
    \CDsystem has an \emph{infinite} maximal trace satisfying a
    $\Cc$-expanding property $\Pp$ given by a B\"uchi
    automaton.
  \end{lemma}

  \begin{proof}
    Applying again Theorem~\ref{th:transP} and slightly modifying the
    \CDsystem we can reduce the satisfaction of $\Pp$ to an instance
    of the repeated reachability problem.
%
%
Observe also that the repeated reachability problem is a special case of
our problem. 
With this reduction, \PSPACE{}-completeness  follows
from Theorem~\ref{thm:liveness}. 

 Let the pushdown system for $\Dd$ be
$\Aa_D=\struct{\Qd,\Sd,\Gd,\dd,\qid,\Aid}$. 
 By taking the product of $\Dd$ with a B\"uchi automaton for
 $\Pp_D$, we can also assume that $\Dd$ has a distinguished set $R$
 of repeating (control) states such that an infinite run of the
 \CDsystem satisfies $\Pp$ iff $\Dd$ visits $R$ infinitely often.

We  add new states and
transitions to $\Aa_D$, so that the leader will signal visits to $R$ by 
writing a special symbol $\#\not\in G$.
We set $G'=G\cup\set{\#}$, and $\Qd' =\Qd \cup \hat \Qd$, where 
$\hat \Qd =\set{\hat q \mid q \in \Qd}$ is a copy of $\Qd$.  The
stack alphabet is unchanged, and we add the following transitions to
$\Aa_D$:
\begin{enumerate}
\item $r \act{a} \hat q$ for every $r \act{a} q$ with
  $r \in R$,
\item $\hat q_1 \act{r(g)} \hat q_2$ for every $q_1 \act{r(g)} q_2$,
\item $ \hat q_1 \act{\dw(\#)} \act{\dw(g)} q_2$ for every $q_1
  \act{\dw(g)} q_2$.
\end{enumerate}
Note that $\dw(\#)$ does not restrict runs of the original system, and does not add new behaviours:
the value $\#$ cannot be read by contributors, and it is immediately followed
by original writes of the leader. 
Since on every infinite run the leader does infinitely often writes, she
will write $\#$ infinitely often iff she visits infinitely often a
state from $R$.
  \end{proof}

  \begin{lemma}\label{l:property-infinite-LTL}
    It is \EXPTIME-complete to decide whether a given pushdown
    \CDsystem has an \emph{infinite} trace satisfying some 
    $\Cc$-expanding property $\Pp$, that is given 
    by an LTL formula.
  \end{lemma}

\begin{proof}
  The lower bound comes from the situation where there are no contributors
  at all \cite{BEM97}.  

  For the upper bound: from an LTL formula
  we first construct a B\"uchi automaton of exponential size for
  $\Pp$.  As in Lemma~\ref{l:property-infinite}, the first step is to
  reduce the problem of deciding if the \CDsystem has a trace in $\Pp$
  to a repeated reachability problem in some \CDpsystem. The leader
  $\Dd'$ there is of exponential size, and $\Cc'$ is of polynomial
  size. 

  As a second step we adapt the procedure given in
  the proof of Theorem~\ref{thm:liveness1}: we do not build the downward closure of
  the leader, but enumerate all
   possible sequences $\n(h_1), \ldots, \n(h_m)$ and intermediate states, 
   instead of guessing them. Then we follow the lines of the proof of
   Theorem~\ref{thm:liveness1}, checking emptiness of pushdowns of
   exponential size (in  \EXPTIME).


  
  First, there are exponentially many possible values for tuples of the form
  $(h_1, \ldots, h_m, p_1, \ldots, p_m, q, A, g)$ where $m \le |G|$,
  $h_1, \ldots, h_m$ is a sequence of pairwise distinct values from $G$,
  $p_1, \ldots, p_m$ are states of $(\Ccf')^\k$, $q$ is a control state of
  $(\Dd')^\k$, $A$ is a stack symbol of $(\Dd')^\k$, and $g \in G$.

  Then, for each such tuple, we can check in (deterministic) exponential
  time if there exists a $(h_1,\ldots,h_m)$-word $v \in \S_{\Dd,\n}^*$ with
  $\lst(v) = g$ such that $(\es,qA,g) \xra{v} (K,qA\a',g)$ in $(\Dd')^\k$,
  for some $\a'$, and $v$ is $\omega$-supported from $p_1, \ldots, p_m$.
  As in Theorem~\ref{thm:liveness1}, we construct for every $1 \le i \le m$
  a finite automaton $\Aa_i$   accepting the projection over $\S_{D,\n}$ of 
  the words $u \in \S_{C,D,\n}^*$ 
  of the form
  \[
  u = u_1 \nu(h_1) \cdots u_i \nu(h_i) \boldsymbol{\Cw(h_i)}  
  u_{i+1} \cdots u_m\nu(h_{m})u_{m+1}
  \]
  and such that $(\es, p_i, g) \xrightarrow{u} (K,p_i,g)$ is a trace in 
  $(\Ccf')^\k$.
  Let $\Aa$ be a pushdown automaton accepting the set of 
  $(h_1, \ldots, h_m)$-words $v$ such that $(\es,qA,g) \xra{v} (K,qA\a,g)$
  in $(\Dd')^\k$.
  To decide the existence of an $\omega$-supported trace satisfying the above
  conditions, we construct $\Aa \cap \Aa_1 \cap \cdots \cap \Aa_m$, which
  is a pushdown automaton of exponential size, and test whether its language
  is empty.

  The \PSPACE\ procedure for the reachability problem described in \cite{LMW15} 
  is very similar, and we can adapt it in the same way to decide in exponential 
  time if there exists a configuration $(M,qA\a,g)$ satisfying 
  $M \ge [p_1, \ldots, p_m]$ that is reachable in the $(\Ccf',\Dd')$-system.

  By Lemma~\ref{lem:Clsupp}, this gives us an \EXPTIME\ algorithm to decide
  if the original \CDsystem has a trace in $\Pp$.
  \end{proof}


\section{Simplifying \CDsystems}
\label{sec:simplify}

In this section, 
we show that a \CDsystem can be simulated by another \CDpsystem such that
all actions in the original system are reflected in  leader writes in
the new system.
So in the \CDpsystem all behaviours of the system, up to stuttering,
will be reflected in the actions of the leader. 

The idea is that in the \CDpsystem the register of the \CDsystem
becomes part of the leader state. 
This releases  the actual register of the \CDpsystem to be used to communicate
about contributor actions. 
Contributors will write into the register the command they want to
perform, and the leader will execute the command and confirm it by
writing  back into the
register.
This confirmation is read by contributors who at this point know that
their request has been read and executed.
For symmetry the leader is also writing the commands she performs to
the register (although they are never read by anybody).
So the set of \emph{register values} of the \CDpsystems~is:
\begin{equation}\label{eq:Gp}
  G'=\set{\qcr g,\qcw g,\acr g,\acw g,\gdr g,\gdw g : g\in G} 
  \cup \set{g'_\init}
\end{equation}
The alphabets of $\Cc'$ and $\Dd'$ are defined as usual:
\begin{equation*}
  \Scp=\set{\Cr(g'),\Cw(g') : g'\in G'}\ \qquad 
  \Sdp=\set{\dr(g'),\dw(g') : g'\in G'}\, .
\end{equation*}

The states of $\Cc'$ are 
\begin{align*}
S' & = S \cup \set{[s,a,s'] : \text{$s,s' \in S$, and $a=\acr g$ or
     $a=\acw g$ for some $g\in G$}}\,.
\end{align*}
The new states of the form $[s,a,s']$ represent the situation when the
contributor has declared that he wants to do the transition $s \act{a}
s'$ and move to $s'$.  In order to really move to $s'$, he needs to
wait for a confirmation from the leader that the action $a$ has been
taken into account.  This mechanism is captured by the following transitions of
$\Cc'$:
\begin{align*}
s\xra{\Cw(\qcw g)} [s,\acw g,s'] \xra{\Cr(\acw g)} s' &
\qquad\text{if}\quad s\xra{\Cw(g)}s' \text{ in } \Cc\\
s\xra{\Cw(\qcr g)} [s,\acr g,s'] \xra{\Cr(\acr g)} s' &
\qquad\text{if}\quad s\xra{\Cr(g)}s' \text{ in } \Cc
\end{align*}

The states of $\Dd'$ are:
\begin{align*}
  T' & =  \set{[t,x] : \text{$t\in T$, and $x=g$, $x=\acw g$, or $x=\acr
       g$, for some $g\in G$}}\,,
\end{align*}
where the component $x$ is supposed to store the value of the register of
the \CDsystem being simulated. 
It can also be a read or write operation, when $\Dd'$ is in the
process of confirming a contributor operation. 
The transitions of $\Dd'$ are:
\begin{align*}
& [t,g]\xra{\Dw(\gdw h)}[t',h]
\quad\text{if}\quad t\xra{\Dw(h)}t' \text{ in } \Dd\\
& [t,g]\xra{\Dw(\gdr g)}[t',g]
\quad\text{if}\quad t\xra{\Dr(g)}t' \text{ in } \Dd\\
& [t,g] \xra{\Dr(\qcw h)} [t,\acw h] \xra{\Dw(\acw h)}[t,h]
\quad\text{for all $t,g,h$}\\
& [t,g] \xra{\Dr(\qcr g)} [t,\acr g] \xra{\Dw(\acr g)}[t,g]
\quad\text{for all $t,g$} \, .
\end{align*}
So transitions of $\Dd$ are simply reflected by transitions of $\Dd'$:
the value of the register of the \CDsystem is stored in the state of
$\Dd'$, and the operation being performed is written into the register
of $\Dd'$.
When a request of an operation from a contributor is read then it is
performed on the value stored in the state and the confirmation of
this operation is written into the register of the \CDpsystem.

In order to state the correspondence between traces of the \CDsystem and
that of the \CDpsystem we define some operations on traces. 
The first one transforms a sequence over the alphabet of the
\CDpsystem into a sequence over the alphabet of the \CDsystem.
A sequence $\trans(u)$ is obtained from $u$ by:
\begin{enumerate}
\item removing all operations of contributors and all read operations
  of the leader, and
\item replacing all write operations $\Dw(a)$ of the leader by $a$,
  for example $\Dw(\acr g)$ is replaced by $\Cr(g)$.
\end{enumerate}
So the operation $\trans(u)$ is the sequence of operations that are
written by $\Dd'$ into the register.

The second operation uses \emph{stuttering expansions} of sequences
over the alphabet of \CDsystems wrt.~contributor actions.  Let $u =
u_0 a_0 u_1 a_1 u_2 \cdots$ be a finite or infinite word over $\Sc
\cup \Sd$, with $a_i\in\Sc$ and $u_i \in \Sd^*$ for all $i$ (or $u_i
\in \Sd^\omega$ if $u$ is infinite but the sequence $u_0, u_1, \ldots$
is finite and of length~$i$).  We write $v\in\sexp(u)$ if there exists
a function $f : \Nat \to \Nat^+$ such that $v = u_0 a_0^{f(0)} u_1
a_1^{f(1)} u_2 \cdots$. This operation is required because a single
confirmation by the leader of a contributor's request can satisfy
several identical requests.

\begin{proposition} \label{th:trans}
  Let $(\Cc',\Dd')$ be obtained from a \CDsystem as described above.
  If $u$ is a trace of the \CDsystem then there is a trace $u'$ of the
  \CDpsystem such that $\trans(u')=u$.
  If $u'$ is a trace of \CDpsystem then there is some $u\in\sexp(\trans(u'))$
  that is a trace of the \CDsystem.
\end{proposition}

The proposition above, and then Theorem~\ref{th:transP}, will follow
from the next lemmas. 

\begin{lemma}\label{lem:CDtoPrime}
  If $u$ is a trace of $(\Cc,\Dd)$ then the trace $u'$ obtained by replacing
  \begin{align*}
  \Dw(g) & \quad \text{by} \quad \Dw(\gdw g) \\
  \Dr(g) & \quad \text{by} \quad \Dw(\gdr g) \\
  \Cw(g) & \quad \text{by} \quad \Cw(\qcw g)\,\Dr(\qcw g)\,\Dw(\acw g)\,\Cr(\acw g) \\
  \Cr(g) & \quad \text{by} \quad \Cw(\qcr g)\,\Dr(\qcr g)\,\Dw(\acr g)\,\Cr(\acr g)
  \end{align*}
  is a trace of $(\Cc',\Dd')$.
  If $u$ is finite and there exists a run of the \CDsystem over $u$
  ending in $(M,t,g)$, then there exists a run of the \CDpsystem over $u'$
  ending in $(M,[t,g],a)$, where $a$ is the last action of $u$.
\end{lemma}

\begin{lemma}
  \label{lem:PrimetoCD}
Let $(n[s_\init],[t_\init,g_\init],g'_\init)\xra{u'}(M',[t,x],g')$ 
be a run of the \CDpsystem.
We can construct by induction on the length of $u'$ a run 
$(n[s_\init],t_\init,g_\init) \xra{u} (M,t,g)$ 
of the \CDsystem and a multiset $N\leq M'$ such that:
\begin{enumerate}
\item $u\in\sexp(\trans(u'))$.
\item\label{I:g} If $x\in G$ then $x=g$. If $x$ is of the form $\acr h$ then $h=g$.
\item\label{I:eq} For all $s \in S$,
  \[M(s) = M'(s)  + \sum_{a\in\set{\Cr,\Cw}\times G} 
  \sum_{s' \in S} M'([s,a,s'])-N([s,a,s']) + N([s',a,s])\, .\]
\item\label{I:comp} If $g' = \qcw h$ or $x = \acw h$, then 
  we have the strict inequality \[\sum_{s,s' \in S} M'([s,\acw h,s']) >
  \sum_{s,s' \in S} N([s,\acw h,s']) \, .\] 
  Similarly if $g' = \qcr h$ or $x = \acr h$.
\item\label{I:N} If $g' = \acw h$, 
  then $N([s,\acw h,s']) = M'([s,\acw h,s'])$ for every $s,s'\in S$.
  Similarly if $g' = \acr g$.
\end{enumerate}
\end{lemma}
Let us explain the role of the multi-set  $N$. States of the form
$[s,a,s']$ can be thought of as transitory states of $\Cc$. They are
counted as $s$, except for those in $N$ that are counted as $s'$. 
Conditions \ref{I:comp} and \ref{I:N} can be interpreted as follows:
between actions $\Dr(\qcw h)$ and $\Dw(\acw h)$, there is at least one
contributor in a state $[s,\acw h,s']$ that is counted as being in $s$;
after $\Dw(\acw h)$, all contributors in $[s,\acw h, s']$ are counted as in 
$s'$.

\begin{proof}
  The proof is by induction on the length of $u'$. We have cases
  depending on the last action.
  Given a transition 
  \[ (M'_1, [t_1,x_1],g'_1) \xra{a'} (M'_2, [t_2,x_2],g'_2) 
  \quad\text{in}\quad (\Cc',\Dd') \]
  and $N_1 \le M'_1$, $(M_1,t_1,g_1)$ satisfying the invariants, we define
  $(M_2, t_2, g_2)$ and $N_2 \le M_2'$ satifying the invariants and such that
  \[ (M_1,t_1,g_1) \xra{a} (M_2,t_2,g_2)
  \quad\text{in}\quad (\Cc',\Dd') \]
  where $a \in \sexp(\trans(a'))$ (we may have $a = \e$).
  We first consider all the possible actions of $\Dd'$, and later those of 
  $\Cc'$. When not stated otherwise, we keep $N_2 = N_1$.
  \begin{itemize}
  \item 
    $\begin{aligned}[t]
      (M',[t_1,g],\ast)&\xra{\dw(\gdw h)} (M',[t_2,h],\gdw h)&\qquad\text{is
    simulated by }\\
      (M,t_1,g)&\xra{\dw(h)} (M,t_2,h) 
    \end{aligned}$
  \item
    $\begin{aligned}[t]
      (M',[t_1,g],\ast)&\xra{\dw(\gdr g)} (M',[t_2,g],\gdr g)&\qquad\text{is
    simulated by }\\
      (M,t_1,g)&\xra{\dr(g)} (M,t_2,g)
    \end{aligned}$
  \item         $(M',[t,g],\qcw h)\xra{\dr(\qcw h)}(M',[t,\acw
    h],\qcw h)$  is simulated by no action. 
  \item          $(M',[t,g],\qcr g)\xra{\dr(\qcr g)}(M',[t,\acr
    g],\qcr g)$  is simulated by no action. 

  \item 
    $\begin{aligned}[t]
      (M',[t,\acw h ],\ast)&\xra{\Dw(\acw h)} (M',[t,h],\acw
                           h)&\qquad\text{is simulated by}\\
      (M_1,t,g)&\xra{\Cw(h)^k}(M_2,t,h)
    \end{aligned}$

    Intuitively in this step we do all the writes that waited to be
    done, so:
    \begin{enumerate}
    \item[(i)] $k=\sum_{s,s'\in S} d_{s,s'}$ with $d_{s,s'}=M'[s,\acw
      h,s']-N_1[s,\acw h,s']$; this number says how many contributors
      there are in state $s$ that want to do $\acw h$ and go to
      $s'$. 
      Note that by invariant~\ref{I:comp}, $k>0$.
    \item[(ii)] 
      $M_2(s) = M_1(s) - \sum_{s' \in S} d_{s,s'} + \sum_{s'\in S} d_{s',s}$.
    \item[(iii)] $N$ is updated from $N_1$ to $N_2$ with
      $N_2[s,\acw h,s'] = M'[s,\acw h,s']$ for every $s,s'\in S$, 
      and $N_2[s,a,s'] = N_1[s,a,s']$ if $a \neq \acw h$.
    \end{enumerate}
    The run is well-defined since by invariant~\ref{I:eq}, 
    $M_1(s) \ge \sum_{s' \in S} d_{s,s'}$.
    Invariant~\ref{I:N} is preserved thanks to item $(iii)$.
    A small calculation shows that invariant \ref{I:eq} is preserved:
    \begin{align*}
      M_2(s) & = M'(s) + 
      \sum_{{a\in\set{\Cr,\Cw}\times G}} \sum_{s'\in S}
      M'[s,a,s'] - N_1[s,a,s'] + N_1[s',a,s]
      \\ & ~~~ - \sum_{s' \in S'} d_{s,s'} + \sum_{s' \in S'} d_{s',s} \\
      & = M'(s) + 
      \sum_{a \neq \acw h} \sum_{s'\in S}
      M'[s,a,s'] - N_1[s,a,s'] + N_1[s',a,s]
      \\ & ~~~ + \sum_{s' \in S} M'([s',\acw h,s]) \\
      & = M'(s) + 
      \sum_{{a\in\set{\Cr,\Cw}\times G}} \sum_{s'\in S}
      M'[s,a,s'] - N_2[s,a,s'] + N_2[s',a,s] \, .
    \end{align*}

    \item 
      $\begin{aligned}[t]
      (M',[t,\acr g ],\ast)&\xra{\Dw(\acr g)} (M',[t,h],\acr
                           g)&\qquad\text{is simulated by}\\
      (M_1,t,g)&\xra{\Cr(g)^k}(M_2,t,g)
    \end{aligned}$

    where $k$ and the multi-sets are defined as for writes.

    \item $(M'_1,[t,x],\ast)\xra{\Cw(\qcw h)}(M'_2,[t,x],\qcw h)$
      is simulated by no operation from $(M,t,g)$. 

      We have $M'_2 = M'_1 - [s] + [[s,\acw g,s']]$ for some $s,s'$.
      Invariant~\ref{I:eq} is preserved, as $M$ and $N$ are not modified and:
      \[M'_1(s) + M'_1([s,\acw g,s']) = M'_2(s) + M'_2([s,\acw g,s')] \]
      Moreover, we have $N \le M'_2$ and
      \[M'_2([s,\acw g, s']) > M'_1([s,\acw g,s']) \ge N([s,\acw g, s']) \]
      so invariant~\ref{I:comp} is verified.
    \item  $(M'_1,[t,x],\ast)\xra{\Cw(\qcr h)}(M'_2,[t,x],\qcr h)$ is
      simulated by no operation from $(M,t,g)$.
      The invariant is preserved, as in the previous case.
      
    \item $(M'_1,[t,x],\acw h)\xra{\Cr(\acw h)}(M'_2,[t,x], \acw h)$
      is simulated by no operation; however,
      we need to update $N_1$ to keep the invariant.
      For this we take the state $[s,\acw h,s']$ that changed to $s'$
      while going from $M'_1$ to $M'_2$, and obtain $N_2$ by
      substracting $1$ from $N_1([s,\acw h,s'])$. This is possible
      since by invariant~\ref{I:N}, 
      $N_1([s,\acw h,s'])= M'_1([s,\acw h,s'])$.
      
      Clearly, invariant~\ref{I:N} is preserved. For invariant~\ref{I:eq},
      observe that:
      \begin{align*}
        M'_1([s,\acw h, s']) - N_1([s,\acw h, s'])
        & = M'_2([s,\acw h, s']) - N_2([s,\acw h, s']) \\
        M'_1(s') + N_2([s,\acw h s']) & = M'_1(s') + N_2([s,\acw h s'])
      \end{align*}
      so the equalities at $s$ and $s'$ are still true. 
      The others do not change.
    \item $(M'_1,[t,g],\acr g)\xra{\Cr(\acr g)}(M'_2,[t,g], \acr g)$
      is simulated similarly.
    \end{itemize}
\end{proof}

Unfortunately, the above construction does not preserve run maximality.
First, it introduces deadlocks: a contributor can declare that he wants to do 
a read, but this read turns out to be illegal, or is simply ignored by
the leader; as a result the contributor gets stuck in state $[s,\acr g,s']$.
Moreover, a deadlock $([s_1, \ldots, s_n],t,g)$ in $(\Cc,\Dd)$ does not 
correspond to a deadlock $([s_1,\ldots, s_n],[t,g],g')$ in $(\Cc',\Dd')$, since
a contributor may be unable to execute an action $s \xra{\Cr(h)} s'$ in 
$(\Cc,\Dd)$, but is always able to execute $s \xra{\Cw(\qcr h)} [s,\acr h, s']$ 
in $(\Cc',\Dd')$.

We can modify the \CDpsystem in order to have a correspondence between maximal
runs of the new system and maximal runs of the \CDsystem.

\begin{lemma}
  \label{lem:CDs}
  There is a $(\Cc'',\Dd'')$-system and a register value $\#$ such that:
  \begin{enumerate}
  \item If $u$ is a (maximal) trace of the \CDsystem, then there exists a 
    (maximal) trace $v$ of the $(\Cc'',\Dd'')$-system with no occurence of 
    $\Dw(\#)$ and such that $\trans(v) = u$.
  \item If $v$ is a (maximal) trace of the $(\Cc'',\Dd'')$-system with no 
    occurence of $\Dw(\#)$, then there exists a (maximal) trace $u$ of the 
    \CDsystem such that $u \in \sexp(\trans(v))$.
  \end{enumerate}
\end{lemma}

\begin{proof}
The result already holds for infinite runs with $(\Cc',\Dd')$.
The idea is to modify the \CDpsystem so that the leader can guess the end of
a finite run corresponding to a maximal run of the \CDsystem. Whenever her guess
is wrong, she can detect it, and write some special value $\#$ to the register.

We start by definying $\Dd''$. 
For a value $g\in G$, we call a state $t$ of $\Dd$ a
\emph{$g$-deadlock} if from $t$ there is no outgoing transition
labeled by $\Dr(g)$ or by a write of some value. 
The leader $\Dd''$ is obtained by adding states $t_f$, $t_\#$ and $t'_\#$ to
$\Dd'$, for all states $t$ of $\Dd$, 
and transitions
\begin{align*}
  [t,g] \xra{\Dw(g)} t_f
  & \quad \text{for all $g\in G$, and $g$-deadlock states $t$ of $\Dd$}\\
  t_f \xra{\Dr(a)} t_\# \xra{\Dw(\#)} t'_\#
  & \quad \text{for all } a \in \set{\#} \cup \set{\qcr g, \qcw g \mid g \in G}
\end{align*}

Similarly, $\Cc''$ is obtained by adding to $\Cc'$ states $s_f$,
$s_\#$, $s'_\#$, for every $s \in S$, and transitions:
\begin{align*}
  [s,a,s'] \xra{\Cr(a)} s'_f 
  & \quad \text{if $\Cc$ cannot do a write from $s$} \\
  s_f \xra{\Cr(g)} s_\# & \quad \text{if $\Cc$ can do $\Cr(g)$ from $s$}\\
  [s,a,s'] \xra{\Cr(g)} s_\# 
  & \quad \text{for all }  g \in G\\
  s \xra{\Cr(g)} s_\# 
  & \quad \text{for all } g \in G\\
  s_\# \xra{\Cw(\#)} s'_\#
\end{align*}
Intuitively, an error $\#$ can occur in two situations.
The first is when  the leader has guessed the end of 
the run by moving to a state $t_f$, but some contributor hasn't:
he is not in a state of the form $s_f$.
The second is when the guessed end configuration $([s^1_f, \ldots,
s^n_f],t_f,g)$ does not correspond to a deadlock, i.e., $\Cc$ can read
$g$ from some state $s^i$.

Observe that the \CDbsystem has the same set of infinite traces 
as the \CDpsystem:
in an infinite run, states $t_f, t_\#, t'_\#$, and thus also $s_\#, s'_\#$ are 
never reached, and states $s_f$ can be replaced by $s$.

From a finite maximal run of the \CDsystem, we can construct a finite maximal
run of the \CDbsystem with no occurence of $\Dw(\#)$, 
as in Lemma~\ref{lem:CDtoPrime}, but having the contributors
move to states $s_f$ instead of $s$ in their last transitions, and adding a 
transition $[t,g] \xra{\Dw(g)} t_f$ at the end of the run.

Conversely, consider a finite maximal run with labeling $v$ of the 
$(\Cc'',\Dd'')$-system, with no occurence of $\Dw(\#)$.
The last action of the leader must be some $[t,g] \xra{\Dw(g)} t_f$,
because she can always do a write from a state of the form $[t,g]$.
Due to the transitions we have added in $\Dd'$, the value of the
register in the final configuration cannot be $\qcr h$, $\qcw h$ nor $\#$.
The value is then $g$ because the last action of the leader was $w(g)$
and contributors can write only $\qcr h$, $\qcw h$ or $\#$.
Due to transitions added above, none of the contributors in the \CDbsystem
can be in a state of the form $s \in S$ or $[s,a,s']$. 
So all the contributors are in states of the form $s_f$ such that $s$
is a $g$-deadlock.
By construction, writing $v = v_1 \Dw(g) v_2$ and $([t^1_f,\ldots,t^n_f],t_f,g)$
the end configuration, there is also a run 
$(n[s_\init],[t_\init,g_\init],g_\init) \xra{v_1v_2} 
([s^1, \ldots, s^n], [t,g], g')$ in the \CDpsystem.
By Lemma~\ref{lem:PrimetoCD}, this leads a run
$(n[s_\init],t_\init,g_\init) \xra{u} ([s^1, \ldots, s^n], t,g)$ in the \CDsystem,
for some $u \in \sexp(\trans(v))$. 
It is maximal since from none of $t$, $s^1,\dots,s^n$ it is possible
to do a write or a read of $g$.
\end{proof}

We can now prove Theorem~\ref{th:transP}. 

\medskip

\begin{proofof}{Thm.~\ref{th:transP}}
For $\tPp$ we take the property
\begin{equation*}
  \tPp=\set{u'' : \text{$\trans(u'')\in\Pp$ and there is no $\Dw(\#)$ in $u''$}}
\end{equation*}
If $\Pp$ is regular, so is $\tPp$. Similarly if $\Pp$ is defined by an
LTL formula. 

Let $(\tCc,\tDd) :=(\Cc'',\Dd'')$ as in Lemma~\ref{lem:CDs}.
If the \CDsystem has a maximal trace $u \in \Pp$, by Lemma~\ref{lem:CDs}, the
\CDtsystem has a maximal trace $u''$ with no occurence of $\Dw(\#)$
and such that $\trans(u'') = u$.

If the $(\tCc,\tDd)$-system has a maximal trace $u''$ such that 
$u'' \in \tPp$, then by Lemma~\ref{lem:CDs} and since $u''$ contains no
occurrence of $\Dw(\#)$, there exists a maximal trace $u$ in the \CDsystem
such that $u \in \sexp(\trans(u''))$.
We have $\trans(u'') \in \Pp$, by the definition of $\tPp$.
Since $\Pp$ is $\Cc$-expanding, $u$ is also in $\Pp$.
\end{proofof}

A further byproduct of the simulation technique used in
Theorem~\ref{th:transP} is that we can simulate a \CDsystem\ with $m$
shared registers by one with a single shared register:

\begin{theorem}\label{th:regs}
    Let $m$ be fixed. For every \CDsystem with $m$ registers, and every $\Cc$-expanding
  regular (resp.~LTL) property $\Pp \subseteq (\Sc \cup \Sd)^\infty$,
  there exists a \CDtsystem with one register and a regular (resp.~LTL)
  property $\tPp \subseteq (\tSd)^\infty$, where $\tSd$ is the input alphabet
  of $\tDd$, such that:
  \begin{quote}
    the \CDsystem has a maximal trace $u \in \Pp$ iff the \CDtsystem has a
    maximal trace $u'$ whose projection on $\tSd$ is in $\tPp$.
  \end{quote}
  Moreover, the \CDtsystem and the property $\tPp$ are
  effectively computable in polynomial time. 
\end{theorem}


\section{Conclusion}

We studied verification questions for parametrized, asynchronous,
shared-memory pushdown systems consisting of a leader process and
arbitrarily many, anonymous contributor processes, as in Hague's
model~\cite{hag11fsttcs,EGM16}. First, we answered an open question of \cite{DEGM15}, by
showing that the complexity of checking liveness in this model is
\PSPACE-complete.  Then we established the complexity of checking
universal reachability.  The developed techniques allowed us to
consider a more general problem, of verifying regular properties that
refer to both leader and contributors, but are stutter-invariant
w.r.t.~contributor actions. We have shown that this problem is
decidable and \NEXPTIME-complete.


\bibliographystyle{abbrv}
\bibliography{biblio}

\newpage
\appendix

\section{Proof of Lemma~\ref{lem:Ccf}}
\label{sec:bounded-contributors}

\begin{proofsketch}
  The idea is to ``distribute'' a run of a contributor into several runs
  with smaller stacks.
  
  Consider a finite run $\rho$ of some copy of $\Cc$ in which the effective stack-height goes
  above $N$. Let $s$ be one of the configurations in the run with effective 
  stack-height greater than $N$. We write $A_1 \cdots A_nA_{n+1}A_{n+2}$
  its effective stack, and $A_1 \cdots A_n \a$ its total stack.
  All symbols of the effective stack, except possibly $A_{n+2}$, are
  eventually popped in the run. In particular, $A_1, \ldots, A_n$ are popped
  strictly before the last action in the run.
  We consider the positions just after the symbols $A_1, \ldots, A_n$ are 
  last pushed before reaching configuration $s$ (resp.~popped after $s$).
  So we have:
  \begin{multline*}
  \rho : (\qic \Aic) \xra{u_n} 
  (p_nA_n\a) \xra{u_{n-1}} (p_{n-1}A_{n-1}A_n\a) \xra{u_{n-2}} \\ \cdots
  \xra{u_1} (p_1A_1 \cdots A_n\a)
  \xra{v_1} (r_1A_2\cdots A_n\a) \xra{v_2} \cdots \xra{v_n} (r_n\a)
  \xra{v} (p_f\a_f)
  \end{multline*}
  where $u_1, \ldots, u_n, v_1, \ldots, v_n, v \in \Sc^+$, and
  in the part $(p_iA_i\cdots A_n\a) \xra{u_{i+1} \cdots u_1} (p_1A_1 \cdots A_n\a)
  \xra{v_1 \cdots v_i} (r_iA_{i+1} \cdots A_n)$, the bottom $A_i \cdots A_n \a$ 
  of the stack is never modified except by the last action, which pops $A_i$.

  Since $n > 2|\Qc|^2|\Gc|$, there must be three indices 
  $1 \le i < j < k \le n$ such that 
  \[(p_i,A_i,r_i) = (p_j,A_j,r_j) = (p_k,A_k,r_k) \, . \]
  We can then construct two smaller runs $\rho_1$ and $\rho_2$ of
  $\Cc$, by removing respectively the parts $u_{j-1} \cdots u_{i}$ and 
  $v_{i+1} \cdots v_{j}$, or $u_{k-1} \cdots u_{j}$ and $v_{j+1}
  \cdots v_{k}$, from $\r$.
  That is,
  \begin{multline*}
    \rho_1: (\qic \Aic) \xra{u_n \cdots u_{j}} (p_jA_j \cdots A_n\a) 
    \xra{u_{i-1} \cdots u_1} (p_1A_1 \cdots A_i A_{j+1} \cdots A_n\a) \\
    \xra{v_1 \cdots v_i} (r_j A_{j+1} \cdots A_n\a)
    \xra{v_{j+1} \cdots v_n v} (p_f\a_f)
  \end{multline*}
  and
  \begin{multline*}
    \rho_2: (\qic \Aic) \xra{u_n \cdots u_{k}} (p_kA_k \cdots A_n\a)
    \xra{u_{j-1} \cdots u_1} (p_1A_1 \cdots A_j A_{k+1} \cdots A_n\a) \\
    \xra{v_1 \cdots v_j} (r_j A_{k+1} \cdots A_n\a)
    \xra{v_{k+1} \cdots v_n v} (p_f\a_f) \, .
  \end{multline*}
  Observe that 
  \begin{itemize}
  \item any transition in $\rho$ can always be associated with a transition of
    $\rho_1$, a transition of $\rho_2$, or both.
  \item since $v \neq \e$, both $\rho_1$ and $\rho_2$ end in the same
  configuration as $\rho$.
  \end{itemize}

  \medskip

  In a run of a \CDsystem, a contributor executing $\rho$ can thus be
  replaced by two contributors executing respectively $\rho_1$ and
  $\rho_2$. They progress together as in the original run on the
  common parts, and for the other parts, one of the two contributor waits
  while the other performs the actions of the original run.
  We repeat this until no contributor ever uses an effective stack of height
  greater than $N$: at each step, we replace one contributor by two
  contributors, but performing strictly shorter runs, so this procedure
  terminates.
\end{proofsketch}

\section{Proof of Lemma~\ref{lem:Clsupp}}
\label{sec:proof-liveness}

  For the left to right direction, we apply Lemma~\ref{lem:loop} to obtain
  a run of the $(\Ccf,\Dd)$-system of the form $(M,qA,g) \xra{u} (M,qA\a',g)$,
  where $(M,qA\a,g)$ is reachable for some $\a$, and $u$ is of the form
  $u = u' \dr(g)$ or $u = u' \dw(g)$, containing some occurrence of 
  $\top$.
  
  We write $M = [p_1, \ldots, p_n]$.
  There are words $v_0 \in \Sd^*$, $v_1, \ldots, v_n \in \Sc^*$
  such that 
  \begin{align*}
    & u \in v_0 \shuffle v_1 \shuffle \cdots \shuffle v_n \quad (v
    \text{ is a shuffle of } v_0,\dots,v_n) , \\
    & qA \xra{v_0} qA\a' \text{ in } \Dd \, ,
  \end{align*}
  and for some permutation $\sigma$ of $\set{1,\ldots,n}$, 
  for all $1 \le i \le n$, 
  \[
  p_i \xra{v_i} p_{\sigma(i)} \text{ in } \Ccf \, .
  \]
  By repeating the run if necessary, we can assume that $\sigma$ is the
  identity. More precisely, for $k$ such that $\sigma^k$ is the identity,
  we can replace $u$ by $(u)^k$, $\a'$ by $(\a')^k$, and $(v_i)$ by $(v_i)^k$ 
  for all $i$.
  
  We let $\overline u$ be the word obtained by replacing each first occurrence 
  of $\Cw(h)$ in $u$ by $\n(h)\Cw(h)$, and $v = \overline u|_{\S_{D, \n}}$.
  Note that $\lst(v) = g$, and $v|_{\Sd} = v_0$.
  Using the fact that $qA \xra{v_0} qA\a'$ in $\Dd$, we can show by induction 
  on the length of $v$ that there is a run of the form 
  $(\es,qA,g) \xra{v} (K,qA\a',g)$ in $\Ddk$. 
  All we need to show is that each read in $v$ is enabled, i.e. that for all
  prefix $v'\dr(h)$ of $v$, either $\lst(v') = h$ or $\n(h)$ occurs in $v'$.
  Consider a prefix $v'\dr(h)$ for which $\lst(v') \neq h$. Then in the
  corresponding prefix $\overline u' \dr(h)$ of $\overline u$ 
  (i.e. with $\overline u'|_{\S_{D,\n}} = v'$), there must be an occurrence of 
  $\Cw(h)$ in $\overline u'$, and the first such occurrence is preceded by 
  $\n(h)$. So $\n(h)$ also occurs in $v'$.

  Denote by $h_1, \ldots, h_m$ the values $h$ such that $\Cw(h)$ occurs
  in $u$, ordered according to their first occurrences.
  For all $1 \le j \le m$, we let $i_j$ be the index such that the first
  occurrence of $\Cw(h_j)$ in $u$ comes from $v_{i_j}$.
  We now show that $v$ is \Clsupp from $(p_{i_1}, \ldots, p_{i_m})$.
  
  For all $1 \le j \le m$, we must find a word $u^j$ such that
  \begin{align*}
    & u^{j} \in \left( \S_{C,D,\n}^*\S_{D,\n} \right) \cap 
          \left( \S_{C,D,\n}^* \n(h_i)\Cw(h_i) \S_{C,D,\n}^* \right) \\
    & u^{j}|_{\S_{D,\n}} = v \\
    & (\es,p_{i_j},g) \xra{u^j} (\set{h_1,\ldots,h_m},p_{i_j},g) 
          \text{ in } \Ccfk \, .
  \end{align*}
  We let $u^j$ be the restriction of $\overline u$ to positions coming
  from $v_0$, $v_{i_j}$ and positions of the $\n(h_k)$.
  By definition, $u^{j} \in \S_{C,D,\n}^* \n(h_i)\Cw(h_i) \S_{C,D,\n}^*$
  and $u^{j}|_{\S_{D,\n}} = v$.
  Moreover, since $u$ ends with $\dw(g)$ or $\dr(g)$, so does $u^j$,
  and $\lst(u^j) = g$.
  We show that $(\es,p_{i_j},g) \xra{u^j} (\set{h_1,\ldots,h_m},p_{i_j},g)$ in 
  $\Ccfk$ similarly to what we did for $v$, using the fact that
  $u^j|_{\S_{\Cc}} = v_{i_j}$ and $p_{i_j} \xra{v_{i_j}} p_{i_j}$ in $\Ccf$.

  \bigskip

  For the right to left direction, consider a reachable configuration
  $(M,qA\a,g)$ and a word $v= v_1\n(h_1) \cdots v_m\n(h_m) v_{m+1}$ satisfying
  conditions~(\ref{item:b}) and~(\ref{item:c}) of
  Lemma~\ref{lem:Clsupp}, and let $\r$ be the run in
  condition~(\ref{item:b}). 
  For all $i$, there exists a word
  
\[u^i = u^i_1\n(h_1) \cdots u^i_i\n(h_i)\Cw(h_i) \cdots u^i_m \n(h_m)
u^i_{m+1}
\]
  such that $u^i_k|_{\S_{\Dd, \n}} = v_k$, and $\Ccfk$ has a run $\rho^i$ 
  of the form
  \begin{multline*}
    (\es, p_i, g) 
    \xra{u^i_1 \n(h_1) \cdots u^i_i \n(h_i)} (\{h_1, \ldots, h_i\}, p'_i, h_i)
    \xra{\Cw(h_i)} (\{h_1, \ldots, h_i\}, p''_i, h_i) 
    \\
    \xra{u^i_{i+1}\n(h_{i+1}) \ldots u^i_{m+1}} (\{h_1, \ldots, h_m\}, p_i, g) \, .
  \end{multline*}

  We denote by $n_i$ the total number of capacity reads of value $h_i$
  occurring either in $\rho$ or in one of the
  $\rho^j$. We show the following lemmas, that
  describe how to construct an ultimately periodic run of the
  $(\Ccf,\Dd)$-system. The proofs of these lemmas ressemble the one
  of Lemma~4 in \cite{LMW15}, but are more involved since we are not
  only interested in reachability and we need to track precisely the
  number of contributors in  a given state.

  \begin{lemma}
    \label{lem:run1}
    There is a run of the $(\Ccf,\Dd)$-system of the form
    \[
    \textstyle
    \left(\sum_{i=1}^m (n_i + 1)\cs{p_i}, q A, g \right) 
    \xra{*}
    \left(\sum_{i=1}^m \cs{p_i} + n_i\cs{p''_i}, q A \a', g \right) \, .
    \]
  \end{lemma}

  \begin{lemma}
    \label{lem:run2}
    There is a run of the $(\Ccf,\Dd)$-system of the form
    \[
    \textstyle
    \left(
    \sum_{i=1}^m \left((n_i + 1)\cs{p_i} + n_i\cs{p''_i}\right), q A, g 
    \right) \xra{u} 
    \left(\sum_{i=1}^m (n_i + 1)\cs{p_i} + n_i\cs{p''_i}, q A \a', g \right)
    \]
    for some $u \in \S_{\Cc,\Dd}^*\top\S_{\Cc,\Dd}^*$.
  \end{lemma}

  Lemmas~\ref{lem:run1} and \ref{lem:run2} show that there is a B\"uchi run 
  starting from $\left(M', q A \a, g \right)$ in the $(\Ccf,\Dd)$-system, 
  for any $M' \ge \sum_{i=1}^m (2n_i +1) \cs{p_i}$.
  Since $(M,qA\a,g)$ is reachable in the $(\Ccf,\Dd)$-system and
  $M \ge [p_1,\ldots,p_n]$,
  by dupplicating the runs of the contributors ending in $p_1,\ldots,p_n$,
  we obtain that the configuration $\left(M', qA \a, g \right)$ where
  $M' = M + \sum_{i=1}^m 2n_i \cs{p_i}$ is also reachable. 
  Hence, the $(\Ccf,\Dd)$-system has a B\"uchi run.

\medskip
  
  \begin{proofof}{Lemma~\ref{lem:run1}}
  The run is obtained as follows. The leader behaves as in $\rho$,
  and for all $i$, one of the contributors, synchronized with the leader,
  behaves as in $\rho^i$. We call this contributor the \emph{main copy}
  of $\rho^i$.
  The remaining contributors starting in $p_i$ follow the main copy of $\rho^i$
  up to reaching $p'_i$, and then stop. Then, each time some process needs to 
  read the value $h_i$, one of the contributors waiting in state $p'_i$
  takes the transition $p'_i \xra{\Cw(h_i)} p''_i$.
  This is defined more precisely below.

  \smallskip

  First, we introduce some notations.
  For any $a \in \S_{\Dd,\Cc,\n}$, and $n \in \Nat$, 
  we denote by $\nst a n$ the word consisting of $n$ $a$'s, 
  and for any word $w = a_1 \cdots a_j$, we let 
  $\nst w n = \nst {a_1} n \cdots \nst {a_j} n$. Recall that $v= v_1\n(h_1)
  \cdots v_m\n(h_m) v_{m+1}$ satisfies conditions~(\ref{item:b}) and~(\ref{item:c}) of
  Lemma~\ref{lem:Clsupp}.
  
  For $1 \le k \le m+1$, we write
  \begin{align*}
    v_k & = \boldsymbol{a_{k,1} \cdots a_{k,\ell_k}} \, , \text{ and} \\
    u^i_k & = x^i_{k,1} \boldsymbol{a_{k,1}} \cdots x^i_{k,\ell_k}
    \boldsymbol{a_{k,\ell_k}} x^i_{k,\ell_k+1} \, \quad (x^i_{k,j}
    \in \S^*_\Cc)\,.
  \end{align*}
  We write $x^i_{k,j}$ as $x^i_{k,j} = y^i_{k,j} z^i_{k,j}$, 
  where $y^i_{k,j}$ is the largest prefix of $x^i_{k,j}$ that consists of
  register reads only.
  Since a register read is necessarily a read from the initial value or a
  read from a value written by the leader, it can only follow an action
  of the leader or another register read.
  Hence each $z^i_{k,j}$ contains only writes and capacity reads.

  \smallskip

  We are going to define a trace of the $(\Ccf,\Dd)$-system 
  as a shuffle of $v|_{\Sd}=v_1 \cdots v_{m+1}$ and each of the 
  $\left(\nst {(u^i_1 \cdots u^i_i)} {n_i+1} \Cw(h_i)
  (u^i_{i+1} \cdots u^i_{m+1})\right)|_{\Sc}$ and
  $\nst {\Cw(h_i)} {n_i}$.
  This corresponds to the intuition that until the first $\Cw(h_i)$,
  all $(n_i+1)$ contributors starting in $p_i$ follow the main copy of $\rho^i$,
  and then the main copy continues alone.
  To make sure that all reads are enabled, the trace will be constructed as
  follows: after each action $\boldsymbol{a_{k,j}}$ of the leader, we put first all
  register reads that follow it in one of the $u^i$ (i.e., actions from some
  $y^i_{k,j+1}$), then all writes or capacity reads (i.e., actions from some
  $z^i_{k,j+1}$).

  \smallskip

  For all $1 \le k \le m$, define
  \begin{align*}
   w_k =~ &
   y^1_{k,1} \cdots y^{k-1}_{k,1} \  
   \nst {(y^k_{k,1})} {n_k+1} \cdots \nst {(y^m_{k,1})} {n_m+1} 
   \\ & \qquad 
   z^1_{k,1} \cdots z^{k-1}_{k,1} 
   \nst {(z^k_{k,1})} {n_k+1} \cdots \nst {(z^m_{k,1})} {n_m+1} 
   \boldsymbol{a_{k,1}}
   \\ & \dots \\
   & y^1_{k,\ell_k} \cdots y^{k-1}_{k,\ell_k} 
   \nst {(y^k_{k,\ell_k})} {n_k+1} \cdots \nst {(y^m_{k,\ell_k})} {n_m+1} 
   \\ & \qquad
   z^1_{k,\ell_k} \cdots z^{k-1}_{k,\ell_k} 
   \nst {(z^k_{k,\ell_k})} {n_k+1} \cdots \nst {(z^m_{k,\ell_k})} {n_m+1} 
   \boldsymbol{a_{k,\ell_k}}
    \\
   & y^1_{k,\ell_k+1} \cdots y^{k-1}_{k,\ell_k+1} 
   \nst {(y^k_{k,\ell_k+1})} {n_k+1} \cdots \nst {(y^m_{k,\ell_k+1})} {n_m+1}
   \\ & \qquad
   z^1_{k,\ell_k} \cdots z^{k-1}_{k,\ell_k} 
   \nst {(z^k_{k,\ell_k+1})} {n_k+1} \cdots \nst {(z^m_{k,\ell_k+1})} {n_m+1}
   \Cw(h_k)\, ,
  \end{align*}
  and similarly $w_{m+1}$, except we remove the last $\Cw(h_k)$.
  We let $w = w_1 \cdots w_{m+1}$.

  We write now $w = b_1 \cdots b_r$, where each $b_{\theta}$ is either one of the
  $\boldsymbol{a_{k,j}}$, a first occurrence of $\Cw(h_i)$,
  a single letter ``$a$'' of one of the $y^i_{k,j}$, $z^i_{k,j}$, 
  or a repetition $\nst a {n_i+1}$ of a letter ``$a$'' in one of the 
  $\nst {(y^i_{k,j})} {n_i+1}$, $\nst {(z^i_{k,j})} {n_i+1}$.
  We define $\bar b_\theta$ as follows:
  \begin{itemize}
  \item If $b_\theta = \Cr(h)$ and $\Cr(h)$ corresponds to a capacity read
    in some $u^i$, then $\bar b_\theta = \Cw(h) \Cr(h)$.
  \item If $b_\theta = \nst {\Cr(h)} n$ and $\Cr(h)$ corresponds to a  capacity read in some
    $u^i$, then $\bar b_\theta = \Cw(h) \nst {\Cr(h)} n$.
  \item If $b_\theta = \dr(h)$ and $\dr(h)$ is a capacity read in $v$, 
    then $\bar b_\theta = \Cw(h) \dr(h)$.
  \item Else, $\bar b_\theta = b_\theta$.
  \end{itemize}
  We let $\bar w = \bar b_1 \cdots \bar b_r$.

  We denote by $\theta_i$ the position of the first occurrence of $\Cw(h_i)$ 
  in $w$.

  For all $\theta \ge 0$, we let $u^{i}(\theta)$ be the prefix of $u^i$
  associated with $b_1 \cdots b_{\theta}$, and $(K_\theta, p^i_\theta, g^i_\theta)$ 
  be the configuration reached in $\rho^i$ after reading $u^i(\theta)$.
  Notice that $K_\theta = \{h_1, \ldots, h_j\}$ where 
  $j = \max \{k \mid \theta_k \le \theta\}$, and thus does not depend on $i$.
  We define similarly $v(\theta)$, $q_\theta$, $\alpha_\theta$, and $g_\theta$
  for the leader.

  We let $n_i(\theta)$ be the sum of the number of capacity reads of $h_i$
  occurring in $v(\theta)$, $u^1(\theta), \ldots, u^m(\theta)$.
  
  We claim that 
  $\left(\sum_{i=1}^m (n_i + 1)\cs{p_i}, q A, g \right)
  \xra{\bar w} 
  \left(\sum_{i=1}^m \cs{p_i} + n_i\cs{p''_i}, q A \a', g \right)$
  in the $(\Ccf,\Dd)$-system.
  More precisely, we show by induction on $\theta$ that for all $\theta \le r$,
  the $(\Ccf,\Dd)$-system has a run of the form:
  \[
  \textstyle
  \left(\sum_{i=1}^m (n_i + 1)\cs{p_i}, q A, g \right)
  \xra{\bar b_1 \cdots \bar b_\theta} 
  \left(M_\theta,q_\theta\a_\theta,g'_\theta\right)
  \]
  where 
  \[M_\theta = \textstyle
  \sum_{\{i \mid \theta < \theta_i\}} (n_i+1) \cs{p^i_\theta} ~~+~~
  \sum_{\{i \mid \theta \ge \theta_i\}} \left( 
  \cs{p^i_\theta} + (n_i - n_i(\theta)) \cs{p'_i} + n_i(\theta) \cs{p''_i}
  \right) \, ,
  \]
  and if $b_{\theta}$ is part of one of the $y^i_{k,j}$, 
  $\nst{(y^i_{k,j})}{n_i+1}$, or if it is one of the $a_{k,j}$,
  then $g'_\theta = g_\theta = g^1_\theta = \ldots = g^m_{\theta}$.
  
  The intuition is that at any time, the main copy of $\rho^i$ is in state
  $p^i_\theta$. 
  Before the first write of $h_i$ (i.e. $\theta < \theta_i$), the $n_i$ 
  remaining copies progress with the main copy. After the first occurrence of 
  $\Cw(h_i)$ (i.e. $\theta \ge \theta_i$), the main copy continues alone.
  The other copies are either in state $p'_i$, waiting to write $h_i$,
  or stopped in state $s''_i$ after writing $h_i$. The transition from $p'_i$
  to $p''_i$ happens each time the next $b_\theta$ corresponds to a capacity
  read in one of the $u^j$ or $v$. So after $\bar b_1 \cdots \bar b_\theta$,
  there are $n_i(\theta)$ copies in $p''_i$.

  \smallskip

  Assume that this holds for some $\theta \ge 0$.

  \smallskip \noindent \textbf{Case 1:}
  $\bar b_{\theta +1}$ is the first occurrence of $\Cw(h_j)$ for some $j$,
  i.e. $\theta + 1 = \theta_j$.
  Then we have $p^j_\theta = p'_j$, $p^j_{\theta + 1} = p''_j$
  and $p^i_\theta = p^i_{\theta + 1}$ for all $i \neq j$.
  Moreover, $n_j(\theta) = n_j(\theta + 1) = 0$.
  Thus $M_{\theta+1} = M_\theta - [p'_j] + [p''_j]$,
  and we can complete the run of the $(\Ccf,\Dd)$-system by letting one of 
  the $n_j +1$ contributors in state $p'_j$ take the transition
  $p'_j \xra{\Cw(h_i)} p''_j$.

  \smallskip \noindent \textbf{Case 2:}
  $b_{\theta +1} = \boldsymbol{a_{k,j}}$, i.e $\bar b_{\theta +1} = \dw(h)$, 
  $\bar b_{\theta +1} = \dr(h)$ or $\bar b_{\theta +1} = \Cw(h) \dr(h)$.
  We have $(K_\theta,q_\theta\a_\theta,g_\theta) \xra{a_{k,j}} 
  (K_{\theta + 1} = K_\theta, q_{\theta+1}\a_{\theta+1},g_{\theta+1} = h)$ in $\Ddk$,
  and for all $i$,
  $(K_\theta,p_\theta,g_\theta) \xra{a_{k,j}}
    (K_{\theta + 1} = K_\theta, p_{\theta} = p_{\theta+1},g_{\theta+1} = h)$ in $\Cck$.
  So the second property we have to prove,
  $g_{\theta+1} = g^1_{\theta+1} = \ldots = g^m_{\theta+1}$, is true.
  We only need to show that
  $(M_\theta,q_\theta\a_\theta,g'_\theta) \xra{\bar b_{\theta +1}} 
  (M_{\theta+1},q_{\theta+1}\a_{\theta+1},h)$.
  \begin{itemize}
  \item If $\bar b_{\theta +1} = \dw(h)$, this is immediate.
  \item If $\bar b_{\theta+1} = \dr(h)$, then $\dr(h)$ is a register read in 
    $v$ (and all of the $u^i$), that is, $h \notin K_{\theta} = K_{\theta + 1}$,
    and $g_\theta = g^i_\theta = \ldots = g^m_\theta = h$.
    For all $i$, we have
    $g^i_\theta = \lst(u^i(\theta)) 
    = \lst(y^i_{k,1} z^i_{k,1} a_{k,1} \ldots y^i_{k,j} z^i_{k,j})$
    (the last equality holds assuming $z^i_{k,j} \not=\e$).
    Since $z^i_{k,j}$ only contains writes and reads of values in $K_\theta$,
    we must have $z^i_{k,j} = \e$ for all $i$.
    Then by induction hypothesis, we have
    $g'_\theta = g_\theta = g^1_\theta = \ldots = g^m_\theta = h$.
    Moreover, $M_\theta = M_{\theta+1}$, so we indeed have
    $(M_\theta,q_\theta\a_\theta,h) \xra{\dr(h)} 
    (M_{\theta+1},q_{\theta+1}\a_{\theta+1},h)$.
  \item If $\bar b_{\theta+1} = \Cw(h)\dr(h)$, then $\dr(h_j)$ is a capacity read
  in $v$, and thus must occur after $\n(h_j)$, i.e. $\theta \ge \theta_j$.
  We also have $n_j(\theta+1) = n_j(\theta) +1$, thus
  $M_{\theta+1} = M_\theta - [p'_j] + [p''_j]$.
  So $(M_\theta,q_\theta\a_\theta,g'_\theta) \xra{\Cw(h)} 
  (M_{\theta +1},q_\theta\a_\theta,h) \xra{\dr(h)} 
  (M_{\theta +1},q_{\theta+1}\a_{\theta+1},h)$.
  \end{itemize}

  \smallskip
  \noindent \textbf{Case 3:}
  $b_{\theta+1}$ is part of one of the $y^i_{k,j}$ or $\nst {(y^i_{k,j})} {n_i+1}$.
  Then it corresponds to a register read in one of the $u^i$, sauy
  $i=i_0$. Thus,
  $\bar b_{\theta+1} = \Cr(h)$ or $\bar b_{\theta+1} = \Cr(h)^{n_{i_0}+1}$.
  We first show that $g'_\theta = g_\theta = g^1_\theta = \ldots = g^m_\theta = h$.
  We have $(K_\theta,p^{i_0}_\theta,g^{i_0}_\theta = h) \xra{\Cr(h)} 
  (K_{\theta+1} = K_\theta,p^{i_0}_{\theta+1},g^{i_0}_{\theta+1} = h)$, and $h \notin K_\theta$.
  If $\theta = 0$, then $g^{i_0}_{\theta} = g'_\theta = g = h$.
  If $\theta > 0$, by construction of $w$, 
  $b_\theta$ is either also part of some $y^{i}_{k,j}$ for $i \le i_0$, 
  or $\boldsymbol{a_{k,j-1}}$, or $\Cw(h_{k-1})$ (then $j = 1$).
  The case $b_\theta = \Cw(h_{k-1})$ is in fact impossible: 
  $u^i(\theta)$ would end with $\n(h_{k-1})$, and we would have 
  $g^{i_0}_\theta = h = h_{k-1} \in K_{\theta}$, which contradicts 
  $h \notin K_{\theta}$. So $b_\theta$ is either part of some $y^i_{k,j}$,
  or $\boldsymbol{a_{k,j-1}}$.
  By induction hypothesis, and since $g^{i_0}_\theta = h$, we obtain
  $g'_\theta = g_\theta = g^1_\theta = \ldots = g^m_\theta = h$.

  Since $g^{i}_\theta = g^{i}_{\theta+1}$ for all $i \neq i_0$ and 
  $g_\theta = g_{\theta+1}$, we also have
  $g_{\theta+1} = g^1_{\theta+1} = \ldots = g^m_{\theta+1} = h$.
  It remains to show that
  $(M_{\theta},q_\theta\a_\theta,h) \xra{\bar b_{\theta+1}} 
    (M_{\theta+1}, q_{\theta+1}\a_{\theta+1} = q_\theta\a_\theta,h)$,
  i.e., $M_{\theta} \xra{\bar b_{\theta+1}} M_{\theta+1}$.
  \begin{itemize}
  \item If $\bar b_{\theta+1} = \nst {\Cr(h)} {n_{i_0}+1}$, 
    then $\theta_{i_0} > \theta + 1$, so $M_\theta \ge (n_{i_0}+1)[p^{i_0}_\theta]$ and
    $M_{\theta+1} = M_\theta - (n_i+1)[p^{i_0}_\theta] + (n_i+1)[p^{i_0}_{\theta+1}]$.
    Thus $M_{\theta} \xra{\nst {\Cr(h)} {n_{i_0}+1}} M_{\theta+1}$.
  \item If $\bar b_{\theta+1} = \Cr(h)$, then $\theta_{i_0} < \theta$, and
    $M_{\theta+1} = M_\theta - [p^{i_0}_\theta] + [p^{i_0}_{\theta+1}]$,
    so $M_{\theta} \xra{\Cr(h)} M_{\theta+1}$.
  \end{itemize}
  
  \smallskip \noindent \textbf{Case 4:}
  $b_{\theta+1}$ is part of one of the $z^i_{k,j}$ or $\nst {(z^i_{k,j})} {n_i+1}$.
  This is similar to case 2 (3rd item).
  \end{proofof}

\medskip

  \begin{proofof}{Lemma~\ref{lem:run2}}
    The idea is that contributors starting in $p_i$
  behave as in the run of Lemma~\ref{lem:run1}, 
  while the contributors starting in $p''_i$ wait
  until the main copy of $\rho^i$ reaches $p''_i$, and then follow it for
  the part $p''_i \xra{u^i_{i+1} \cdots u^i_{m+1}} p_i$.
  We do not give all details of the proof, which are very similar to the proof
  of Lemma~\ref{lem:run1}. We only explain how to define the run of the 
  \CDsystem, and state the invariant for the induction.

  We let
  \begin{align*}
   w'_k = \ & 
   \nst{(y^1_{k,1})}{n_1+1} \cdots \nst {(y^m_{k,1})} {n_k+1} 
   \nst {(z^1_{k,1})} {n_k+1} \cdots \nst {(z^m_{k,1})} {n_m+1} 
   \boldsymbol{a_{k,1}}
   \\ & \dots \\
   & 
   \nst{(y^1_{k,\ell_k})}{n_1+1} \cdots \nst {(y^m_{k,\ell_k})} {n_k+1} 
   \nst {(z^1_{k,\ell_k})} {n_k+1} \cdots \nst {(z^m_{k,\ell_k})} {n_m+1} 
   \boldsymbol{a_{k,\ell_k}}
    \\
    &
   \nst{(y^1_{k,\ell_k + 1})}{n_1+1} \cdots \nst {(y^m_{k,\ell_k + 1})} {n_k+1} 
   \nst {(z^1_{k,\ell_k + 1})} {n_k+1} \cdots \nst {(z^m_{k,\ell_k + 1})} {n_m+1} 
   \Cw(h_k) \, ,
   \end{align*}
  and $w'$, $b'_\theta$, $\bar b'_\theta$, $\bar w$ as before.
  One can show by induction on $\theta$ that the $(\Ccf,\Dd)$-system has a run 
  of the form:
  \begin{multline*}
  \textstyle
  \left(\sum_{i=1}^m (n_i + 1)\cs{p_i} + n_i \cs{p''_i} , t A, g \right)
  \xra{\bar w_1 \cdots \bar w_\theta} \\
  \textstyle
  \left(
  \begin{array}{l}
  \sum_{\{i \mid \theta < \theta_i\}} \left(
  (n_i+1) \cs{p^i(\theta)} + n_i \cs{p''_i} 
  \right) + \\
  \sum_{\{i \mid \theta \ge \theta_i\}} \left( 
  (n_i + 1)\cs{p^i(\theta)} + 
  (n_i - n_i(\theta)) \cs{p'_i} + 
  n_i(\theta) \cs{p''_i}
  \right)
  \end{array}, 
  q_\theta \a_\theta, g' \right) \, .
  \end{multline*}
 \end{proofof}


\section{\CDsystems with multiple registers}

The definition of \CDsystems extends to systems with $m$
registers as expected.
Given a set $G$ of register values, we let
\begin{align*}
\Sck & = \set{\Cr_i(g), \Cw_i(g) : 1 \le i \le m, g \in G}\, ,\qquad\text{and}\\
\Sdk & = \set{\Dr_i(g), \Dw_i(g) : 1 \le i \le m, g \in G} \, .
\end{align*}
For instance, $\Cr_i(g)$ corresponds to a contributor read from the $i$-th 
register.

Assume $\Cc$, $\Dd$ are transition systems over $\Sck$ and $\Sdk$, resp.
Configurations of the \CDsystem are described as tuples 
\[(M \in \Nat^S, t \in T, g_1 \in G, \ldots, g_n \in G) \, . \]
Transitions are defined similarly to the case of \CDsystems with one register:
\begin{align*}
  (M,t,g_1,\ldots g_i, \ldots,g_n) \xra{\dw_i(h)} 
  & (M,t',g_1,\ldots h, \ldots,g_n)
  && \text{if $t\xra{\dw_i(h)} t'$ in $\D$}\,, \\
  (M,t,g_1,\ldots g_i, \ldots,g_n) \xra{\dr_i(h)} 
  & (M,t',g_1,\ldots h, \ldots,g_n)
  && \text{if $t\xra{\dr_i(h)} t'$ in $\D$ and $h = g_i$}\,, \\
  (M,t,g_1,\ldots g_i, \ldots,g_n) \xra{\Cw_i(h)} 
  & (M',t,g_1,\ldots h, \ldots,g_n)
  && \text{if $M\xra{\Cw_i(h)} M'$ in $\d$}\,, \\
  (M,t,g_1,\ldots g_i, \ldots,g_n) \xra{\Cr_i(h)} 
  & (M',t,g_1,\ldots h, \ldots,g_n)
  && \text{if $M\xra{\Cr_i(h)} M'$ in $\d$ and $h = g_i$}\, . \\
\end{align*}

\medskip

From a \CDsystem with $m$ registers, we can construct an ``equivalent''
\CDpsystem with $1$ register, using the ideas of Section~\ref{sec:simplify}.
The contents of the $m$ registers will be stored in the states of $\Dd'$, and 
the unique register of the \CDpsystem will be used to communicate about the 
actions of the contributors.
So the values of the register are:
\[G' = \set{\qcri g, \qcwi g, \acri g, \acwi g, \gdri g, \gdwi g \mid 
1 \le i \le m \text{ and } g \in G} \]
and the alphabets of $\Cc'$ and $\Dd'$ are defined as usual:
\[
\Sc' = \{\Cr(a), \Cw(a) \mid a \in G'\}, \qquad
\Sd' = \{\Dr(a), \Dw(a) \mid a \in G'\} \, .
\]
For a word $u \in (\Sc' \cup \Sd')^\infty$, we define as before $\trans(u)$
by removing all actions of the contributors and all reads of the leader, and
replacing leader writes $\Dw(a)$ by~$a$.

The states of $\Dd'$ are:
\begin{align*}
  T' = T \times G^m \cup \{[t,x_1 \ldots, x_n] \mid \ 
    & \text{$x_i = \acri g$ or $x_i = \acwi g$ for some $i$}, \\
    & \text{and $x_i \in G$ for all $j \neq i$}\}
\end{align*}
and the states of $\Cc'$:
\begin{align*}
  S' = S \cup \set{[s,a,s'] \mid s, s' \in S \text{ and $a = \qcri g$ or
    $a = \qcwi g$ for some $i,g$}} \, .
\end{align*}
The transitions of $\Cc'$ are defined as in Section~\ref{sec:simplify}:
\begin{align*}
s\xra{\Cw(\qcwi g)} [s,\acwi g,s'] \xra{\Cr(\acwi g)} s' &
\quad\text{if}\quad s\xra{\Cw_i(g)}s' \text{ in } \Cc\\
s\xra{\Cw(\qcr g)} [s,\acri g,s'] \xra{\Cr(\acri g)} s' &
\quad\text{if}\quad s\xra{\Cr_i(g)}s' \text{ in } \Cc
\end{align*}
and similarly for transitions of the leader, except $\Dd'$ now keeps the 
values of $m$ registers:
\begin{align*}
& [t,g_1,\ldots,g_i,\ldots,g_n]\xra{\Dw(\gdwi h)}[t',g_1,\ldots,h,\ldots,g_n]
\quad\text{if}\quad t\xra{\Dw_i(h)}t' \text{ in } \Dd\\
& [t,g_1,\ldots,g_i,\ldots,g_n]\xra{\Dw(\gdr {g_i})}[t',g_1,\ldots,g_i,\ldots,g_n]
\quad\text{if}\quad t\xra{\Dr_i(g_i)}t' \text{ in } \Dd\\
& [t,g_1,\ldots,g_i,\ldots,g_n] \xra{\Dr(\qcwi h)} 
[t,g_1,\ldots,\acwi h,\ldots,g_n] \xra{\Dw(\acwi h)}
[t,g_1,\ldots,h,\ldots,g_n] \\
& [t,g_1,\ldots,g_i,\ldots,g_n] \xra{\Dr(\qcri {g_i})} 
[t,g_1,\ldots,\acri {g_i},\ldots,g_n] \xra{\Dw(\acri {g_i})}
[t,g_1,\ldots,g_i,\ldots,g_n] \, . \\
\end{align*}

\begin{theorem}

  \begin{enumerate}
  
  \item For every trace $u$ of the \CDsystem, there exists a trace $u'$ of the
    \CDpsystem such that $\trans(u') = u$. 

    If $u$ is finite and and the \CDsystem has a run over $u$ ending in 
    $(M,t,g_1,\ldots,g_n)$, then the \CDpsystem has a run over $u'$ ending in 
    $(M,[t,g_1,\ldots,g_n],a)$ where $a$ is the last action of~$u$.

  \item For every trace $u'$ of the \CDpsystem, there exists a trace
    $u \in \sexp(\trans(u'))$ in the \CDsystem. 
    
    If $u'$ is finite and the \CDpsystem has a run over $u'$ ending in 
    $(M,[t,g_1,\ldots,g_n],a)$ with $M \in \Nat^S$ and $g_1,\ldots,g_n \in G$, 
    then the \CDsystem has a run over $u$ ending in $(M,t,g_1,\ldots,g_n)$.
  \end{enumerate}
\end{theorem}

As in Section~\ref{sec:simplify}, we can also modify the \CDpsystem to
preserve maximality of runs, and prove Theorem~\ref{th:regs}.


\end{document}